\newcounter{myctr}
\def\myitem{\refstepcounter{myctr}\bibfont\noindent\ifnum\themyctr>9\else\phantom{0}\fi\hangindent17pt\themyctr.\enskip}

\documentclass[a4paper, halfparskip,english, draft]{article}
\usepackage[english]{babel}
\usepackage{amsmath, amsthm, amssymb, bbm}
\usepackage{nicefrac}
\usepackage{varioref}

\usepackage[T2A]{fontenc}
\usepackage{hyperref}
\usepackage[super,sort,compress]{cite}
\usepackage{tikz}
\usepackage[utf8]{inputenc}
\usetikzlibrary{arrows}
\usetikzlibrary{plothandlers}
\usetikzlibrary{decorations,decorations.pathmorphing,decorations.markings}
\newtheorem{thm}{Theorem}
\newtheorem{pro}[thm]{Proposition}
\newtheorem{lem}[thm]{Lemma}
\newtheorem{cor}[thm]{Corollary}
\theoremstyle{definition}
	\newtheorem{definition}{Definition}
\theoremstyle{remark}

\begin{document}
\renewcommand{\a}{\alpha}
\renewcommand{\b}{\beta}
\newcommand{\g}{\gamma}
\renewcommand{\d}{\delta}
\newcommand{\e}{\varepsilon}
\newcommand\z{\zeta}
\renewcommand{\th}{\vartheta}
\renewcommand{\k}{\kappa}
\newcommand{\n}{\nu}
\newcommand{\no}{\omega}
\def\oo{\omega}
\renewcommand{\r}{\rho}
\newcommand{\s}{\sigma} 
\newcommand{\ph}{\varphi}
\renewcommand{\t}{\tau}
\newcommand{\Oo}{\Omega}
\newcommand{\Om}{\Omega}

\newcommand\mutilde{\widetilde\mu}
\def\mtilde{\widetilde m}
\newcommand\nutilde{\widetilde\nu}
\newcommand\Mtilde{\widetilde M}
\newcommand\Ntilde{\widetilde N}
\def\Mctilde{\widetilde{\cal M}}
\def\Deltatilde{\widetilde{\Delta}}
\def\Gammatilde{\widetilde{\Gamma}}
\def\Sigmatilde{\widetilde{\Sigma}}
\def\Deltahat{\widehat\Delta}
\def\Mhat{\widehat M}
\def\Mchat{\widehat{\cal M}}
\def\Nhat{\widehat N}
\def\Sigmahat{\widehat\Sigma}
\def\Gammahat{\widehat\Gamma}
\def\muhat{\widehat\mu}
\def\nuhat{\widehat\nu}
\def\Rbar{{\overline\RR}}
\def\toexp{\mathop{\exp}^{\longrightarrow}}
\def\Rp{{\RR\setminus\{0\}}}
\def\intRn{\int_{\RR\setminus\{0\}}}
\def\Qphi{Q_\ph}
\def\eiox{e^{i\oo x}}
\def\ant#1#2{\{#1[\one],#2\}}
\def\Pfi#1{(#1-\ph(#1)\cdot\one)}
\def\ket#1{|#1\rangle}
\def\bra#1{\langle#1|}
\def\supp{{\rm supp}}

\def\A{{\cal A}}
\def\B{{\cal B}}
\def\C{{\cal C}}
\def\D{{\cal D}}
\def\F{{\cal F}}
\def\G{{\cal G}}
\def\H{{\cal H}}
\def\K{{\cal K}}
\def\L{{\cal L}}
\def\M{{\cal M}}

\def\Cbar{\overline{\cal C}}
\def\Cb{C_b}
\def\Ltilde{\tilde{L}}
\def\CP{{\rm CP}}
\def\norm#1{\|#1\|}
\def\Norm#1{\left\|#1\right\|}
\def\norms#1{\|#1\|^2}
\def\Norms#1{\left\|#1\right\|^2}
\def\normfour#1{\|#1\|_4}
\def\trinorm#1{|||#1|||}

\def\sod#1{\sum_{#1=1}^d}
\def\sok#1{\sum_{#1=1}^k}
\def\sokn#1{\sum_{#1=1}^{k_n}}
\def\sol#1{\sum_{#1=1}^l}
\def\som#1{\sum_{#1=1}^m}
\def\son#1{\sum_{#1=1}^n}
\def\sn#1{\sum_{#1=0}^{n-1}}
\def\sm#1{\sum_{#1=0}^{m-1}}
\def\szi#1{\sum_{#1=0}^\infty}
\def\soi#1{\sum_{#1=1}^\infty}
\def\li#1{\lim_{#1\to\infty}}
\def\lid#1{\lim_{#1\downarrow0}}
\def\szn#1{\sum_{#1=0}^n}
\def\intR{\int_{-\infty}^\infty}
\def\intRp{\int_{\RR\setminus\{0\}}}
\def\wlim{\qopname\relax m{w{-}lim}}
\def\ddt{\frac d {dt}}
\def\Mdop{\A^{\rm op}}
\def\dth{{\textstyle\frac 1 d}}
\def\Ld{L^{\textnormal{diff}}}
\def\Lj{L^{\textnormal{jump}}}
\def\Lz{L^{\textnormal{zero}}}
\def\half{{\textstyle{\frac 12}}}
\def\third{{\textstyle{\frac 13}}}
\def\twothirds{{\textstyle{\frac 23}}}
\def\quarter{{\textstyle{\frac 14}}}

\def\Dom{{\rm Dom}}
\def\id{{\rm Id}}
\def\Ad{{\rm Ad}}
\def\Re{\mathop{\rm{Re}}}
\def\Im{\mathop{\rm{Im}}}
\def\tr{{\mathrm{tr}}}

\def\EE{{\mathbb E}}
\def\NN{{\mathbb N}}
\def\ZZ{{\mathbb Z}}
\def\RR{{\mathbb R}}
\def\CC{{\mathbb C}}
\def\tuple#1#2{#1_1,#1_2,\ldots,#1_{#2}}
\def\tupple#1#2{#1_1,\ldots,#1_{#2}}
\def\one{{\mathchoice {\rm 1\mskip-4mu l} {\rm 1\mskip-4mu l}{\rm 1\mskip-4.5mu l} {\rm 1\mskip-5mu l}}}
\def\to{\rightarrow}
\def\tto{\longrightarrow}
\def\ten{\otimes}
\def\inp#1#2{\langle#1,#2\rangle}
\def\Inp#1#2{\left\langle#1,#2\right\rangle}
\def\implies{\Longrightarrow}  
\def\and{\quad\hbox{and}\quad} 
\def\an#1{\quad\hbox{and #1}\quad} 
\def\twovector#1#2{\begin{pmatrix}#1\\#2\end{pmatrix}}
\def\twomatrix#1#2#3#4{\begin{pmatrix}#1&#2\\#3&#4\end{pmatrix}}
\def\twosvector#1#2{\left(\begin{smallmatrix}#1\\#2\end{smallmatrix}\right)}
\def\twosmatrix#1#2#3#4{\left(\begin{smallmatrix}#1&#2\\#3&#4\end{smallmatrix}\right)}
\def\set#1#2{\{#1|\;#2\;\}}


\title{CONTINUOUS OBSERVATION OF QUANTUM SYSTEMS}

\author{HANS MAASSEN}


\maketitle


\begin{abstract}
In a series of papers in the 1980's \cite{Hol1}\cite{Hol2}\cite{Hol3} Alexander Holevo proved a classification theorem for continuous quantum measurement processes,
or, as they would today be called, stationary quantum trajectories in continuous time.
His main tools were functional analytic in character: starting from a Bochner-type inequality he employed dilation techniques for positive definite
kernels. Here we give an alternative, more probabilistic proof: we use weak convergence of measures and employ L\'evy's Continuity Theorem.
We clarify the boundedness conditions in Holevo's theorem, and supply a simple example from quantum optics. 
\end{abstract}



\markboth{Hans Maassen}{Continuous Observation of Quantum Systems}

\section*{Introduction}	
The mathematical formulation of measurement in quantum mechanics has developed considerably since the first formulations of the theory
in the 1920's. 
The first step was the formulation by Max Born of the probability rule.
The state $\psi_t$ of a quantum system, a unit vector in its Hilbert space $\H$, normally develops under a group of unitary operators: $\psi_t=u_t\psi_0$. 
Propositions on a quantum system are described by orthogonal projection operators $p$ on $\H$,
and measurements by n-tuples $\tuple p n$ satisfying $p_1+p_2+\ldots+p_n=\one$.
When a measurement is made at time $t$, a random outcome is obtained according to the distribution
$(\norms{p_1\psi_t}, \norms{p_2\psi_t}, \ldots, \norms{p_n\psi_t})$.
The second step was taken by von Neumann: if we want to know what happens to the system after the measurement, 
we must start its evolution afresh in the state $p_i\psi_t/\norm{p_i\psi_t}$, the normalized projection of $\psi_t$ unto the subspace $p_i\H$ associated to the outcome $i$. 
These two rules, the Born rule and the von Neumann projection postulate,
for several decades dominated the description of quantum measurement.
But in the 1970's it was realized\cite{Hol85}\cite{DavLew70}\cite{Kraus71}\cite{Dav76}
that these rules are too narrow to describe what naturally goes on in laboratories.
Repeated measurements are already problematic, but the scheme fails miserably for measurement in continuous time,
where it leads to the "Zeno Paradox": the stroboscopic action of increasingly frequent von Neumann projections blocks the motion of the quantum system
under observation completely\cite{MisraSudarshan}.

\noindent
In order to make any progress, the measurement had to be "softened": it was proposed by, among others, Ludwig\cite{Lud54} and Holevo\cite{Hol11}
to replace the projections by arbitrary positive operators $\tuple m n$ on $\H$ having sum $\one$,
or even a positive operator-valued measure (POVM) on an infinite space of measurement outcomes.
This is the most general situation if we do not care what happens to the quantum system after the measurement,
but if we wish to preserve the system for further investigation,
then for every possible outcome a prescription is required how the state is to be changed if that outcome occurs.
This is done by a probability measure on the outcome space, taking values in the space of all operations on the quantum system.
Such measures were introduced by Davies\cite{Dav76} under the name of {\it instruments}.
Repeated quantum measurement is described by the iteration of instruments, yielding a sequence of outcomes.
The conditional expectations of propositions $p$, given the observed outcomes until time $t$, constitute a conditional state $\Theta_t$,
which, together with the observations themselves, form a Markov chain, known as a {\it quantum trajectory}. 
This approach was highly succesful, and a whole branch of literature\cite{Gisin84}\cite{Moel93}\cite{Car93} on quantum trajectories,
gave accurate descriptions of many physical experiments.
Their ergodic behaviour has been studied extensively\cite{Cresser} \cite{KMErg} \cite{Benoist} \cite{Benoist2} \cite{Benoist3}.

\noindent
Performing the measurements given by a pair of instruments, and {\it adding} their outcomes, leads to a new instrument, the {\it convolution} of the pair.
Finally, measurement in continuous time is described by a {\it convolution semigroup} of instruments\cite{BarLP83},
the object of study of this paper.
Examples of observation records described in this way are the counting of particles\cite{SrinDav81} \cite{KimbleMandel} by a detector,
or accumulative field measurement\cite{BarGreg09}.
It is these semigroups that are classified by Holevo's Theorem \cite{Hol1}\cite{Hol2}\cite{Hol3}(here Theorem \ref{ThmHolevo}),
which was presented in its final form around the turn of the century\cite{Hol4}\cite{Hol5}.
Outcomes were supposed to lie in $\RR^n$, later even in general Lie groups.\cite{BarHolLup93}
Here, for simplicity we shall consider outputs in $\RR$ only, but the extension to $\RR^n$ is straightforward.

\noindent
(Weakly continuous) convolution semigroups of instruments on the real line correspond to stochastic processes, which,
not unlike the needle of a seismograph, follow the events that are going on inside the quantum system.
The needle may make sudden jumps, indicating the occurrence of a quantum jump,
or jiggle in an irregular fashion like a Brownian particle, reflecting a diffusive motion of the quantum state.
It is the merit of Holevo's Theorem to indicate how every observation in continuous time
can thus be decomposed into jumps and jiggles.

\noindent
In preparing a treatment of Holevo's Theorem for our textbook\cite{Buch} with B. K\"ummerer, we developed the present presentation and proof. 

\medskip\noindent
This paper is structured as follows.
In Section \ref{QuantMeas} the above discussion is specified in more technical terms,
embedding quantum measurement into the category {\bf QP} of C*-algebras with completely positive unital maps\cite{Buch} \cite{KMThesis}, \cite{KMGreifswald}.
In Section \ref{WeakConv} a natural topology on instruments  is introduced: the topology of weak convergence of measures.
Section \ref{LevyCont} describes convolution of  instruments and Paul L\'evy's Continuity Theorem in the present context,
the main work horse of our approach,
leading to the characteristic exponents of convolution semigroups in Section \ref{ConvSem}.
A central example is treated in Section \ref{BasicDavies}, applied to the phenomenon of Resonance Fluorescence in Quantum optics\cite{KimbleMandel}.
In Section \ref{InfDivInstQuantDeFinetti} we give an analogue of De Finetti's characterization of characteristic exponents of convolution semigroups of probability measures on the line,
and connect these generalized exponents to stochastic measurement processes (the "needle") in Section \ref{Needle}.
After a short discussion of linear algebra on "superoperators" in Section \ref{LinAlg} we proceed to formulate and prove Holevo's Theorem in Section \ref{Holevo}.
The theorem provides a generalization, from the L\'evy-Khinchin formula for characteristic exponents of 
classical stochastic processes with stationary independent increments, to Markovian observation processes on the real line.

\noindent
In this new derivation we hit upon an improvement of the theorem in the following sense:
In the L\'evy-Khinchin formula the subtle part is the integrability condition on the jump sizes when an infinity of jumps occurs in a finite time.
In the quantum context {\it two} integrability conditions are required, as is made explicit in Holevo's later review papers\cite{Hol4}\cite{Hol5}.
Holevo's first condition concerns jump sizes in outcome space, like in the L\'evy-Khinchin formula.
We have replaced his, somewhat opaque, second condition by an integrability condition on jump sizes in state space.
In Section \ref{InfDivInstIntCond} we prove that our new condition is equivalent to Holevo's second integrability condition.

\section{Quantum Measurement}\label{QuantMeas}
{\bf The system:}
Let $d\ge1$, and
let $\A$ denote  be the C*-algebra $M_d$ of all complex $d\times d$-matrices, operators on the finite dimensional Hilbert space $\H=\CC^d$,
to be interpreted as the {\it algebra of observables} of a quantum system.
States on the quantum system are described by positive linear functionals $\r:\A\to\CC$ with $\r(\one)=1$.
{\it Events} or {\it propositions} concerning the quantum system are described by orthogonal projections $p=p^*=p^2\in\A$.
Interpretation: if the system is in de state $\r$, the probability for $p$ to hold is $\r(p)$.

\noindent{\bf The outcomes:}
Let $X$ be a compact metric space, the space of possible outcomes of measurements.
($X$ will be specified in Section \ref{Addition}.) 
By $C(X)$ we denote the unital C*-algebra of continuous complex functions on $X$.
Propositions concerning the outcome $x\in X$ are modelled by Borel sets, forming a $\s$-algebra $\B(X)$.
A probability measure $\mu$ on $X$ associates to every Borel set $B\in\B(X)$ its probability $\mu(B)$.
It defines a {\it state} $\mutilde$ on $C(X)$ by
   $$\mutilde(f):=\int_Xf(x)\mu(dx)\;.$$
In fact, by the Riesz Representation Theorem there is a one-to-one correspondence between Radon measures $\mu$ on $X$ and continuous linear
functionals $\mutilde$ on $C(X)$.
This theorem allows an abuse of notation, which we announce already at this early stage: for a Borel subset $B$ of $X$, 
the indicator function $1_B$ will usually not lie in $C(X)$;
nevertheless we will sometimes write $\mu(B)$ as $\mutilde(1_B)$.

\subsection{Destructive quantum measurement: POVM's}\label{POVM}
Let $\A_{+}$ denote the cone of positive elements of $\A$.
By a {\it destructive measurement on the system $\A$ yielding values in $X$} we mean a measure $m:\B(X)\to\A_{+}$ with $m(X)=\one$,
often called {\it positive operator valued measure} (POVM)\cite{Kraus71}.
Interpretation: if the quantum system is in the state $\r$, then an outcome $x\in B$ is obtained with probability $\r(m(B))$.
After the measurement the system need no longer be available.
In particle-counting experiments the particles being counted undergo such a measurement.

\bigskip
\begin{tikzpicture}
\draw (-2,0) node {\phantom{Initial}};
\draw (-0.3,0) node {$\mtilde(f)=\int_X f(x)m(dx) \leftarrow\phantom{xxxxxxxxxxxxxxxxxxxx}$};
\draw (1,0.3) node {$\A$};
\draw[thick,decorate,decoration={snake}] (0,0) -- (2,0);
\draw[thick] (2,-0.5) rectangle (3,0.5);
\draw (2.5,0) node {$\mtilde$};
\draw[thick] (3,0) -- (5,0) node[right] {$\leftarrow f$};
\draw (4,0.3) node {$C(X)$};
\end{tikzpicture}

\noindent
A POVM is called a {\it von Neumann measurement} if all the values $m(B)$ with $B\in\B(X)$ are projection operators.

\noindent
Note that $\A$ and $C(X)$ are unital C*-algebras, and $\mtilde$ is a completely positive map.
It is thus a {\it morphism} in the category {\bf QP} of C*-algebras with unital completely positive maps. 
Diagrams like the above are known as {\it graphical diagrams} in Category Theory\cite{Selinger}.
When read from right to left, such a diagram acts on observables, in the "Heisenberg picture".
When read from left to right, in the "Schr\"odinger picture", which is the more "physical" direction,
it acts on a state $\r$, to yield a probability distribution $\mtilde^*(\r): B\mapsto\r\bigl(m(B)\bigr)$ on $X$ as its output.

\subsection{Non-destructive quantum measurement: instruments}\label{NDQuantMeas}
Let $\L(\A)$ denote the space of all linear maps $\A\to\A$ ("superoperators"),
which we consider as a C*-algebra of operators on the Hilbert space $(\A,\inp\cdot\cdot)$,
where the inner product is given by $\inp a b:=\frac1d\tr(a^*b)$. (Cf. Section \ref{LinAlg}.)
Let $CP(\A)\subset\L(\A)$ be the cone of completely positive operators on $\A$.
\begin{definition}[Instrument]\label{DefInstr}
By a {\it completely positive measure} on $X$ we mean a $\s$-additive function $M$ from $\B(X)$ to $\CP(\A)$.
Such a measure is called an {\it instrument} if $M(X)[\one]=\one$.
\end{definition}

\noindent
Interpretation: after the measurement, described by the instrument $M$, is performed on the quantum system in the state $\r$,
the probability for the outcome to lie in $B\in\B(X)$ and the event $p\in\A$ to occur is $\r\bigl(M(B)[p]\bigr)$.

\noindent
We state without proof a simple extension of the Riesz Representation Theorem.
Since we wish to apply it to $\A$-valued measures as well as to $\L(\A)$-valued measures,
we formulate it in terms of a general finite dimensional normed vector space $E$ over $\CC$.
We shall denote general $E$-valued measures by $\mu$, $\nu$, often switching to capital letters $M$, $N$ if they are $\L(\A)$-valued, or even $\CP(\A)$-valued.

\begin{pro}[Riesz Representation Theorem]\label{Riesz}
Let $X$ be a compact metric space, and $E$ a normed vector space of finite dimension.
There is a one-to-one correspondence between continuous functionals $\mutilde:C(X)\to E$ and $E$-valued measures $\mu$ on $X$, given by
   $$\mutilde(f)=\int_X f(x)\mu(dx)\;.$$
\end{pro}\noindent
In the case that $E=\L(\A)$ and $\mu=:M$ is an instrument,
the map $\Mtilde$ can be viewed as a completely positive map $C(X)\ten\A\to\A:f\ten a\mapsto\Mtilde(f)[a]$.
The normalization of $M$ says that $\Mtilde$ is unital: $\Mtilde(1\ten\one)=\one$.
So again $\Mtilde$ is a morphism in the category {\bf QP} mentioned above.
The positive operator $M(B)[p]$ alluded to above is equal to $\Mtilde(1_B\ten p)$ (here we ignore  the fact that $1_B\notin C(X)$).
In a diagram:

\bigskip
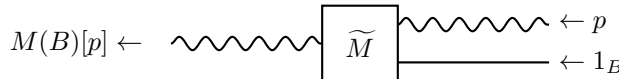
\begin{tikzpicture}
\draw (-2,0) node {\phantom{Initial}};
\draw (-0.3,0) node {$M(B)[p]\leftarrow\phantom{xxxxxxxxx}$};
\draw[thick,decorate,decoration={snake}] (0,0) -- (2,0);
\draw[thick] (2,-0.5) rectangle (3,0.5);
\draw (2.5,0) node {$\Mtilde$};
\draw[thick,decorate,decoration={snake}] (3,.25) -- (5,.25) node[right,black] {$\leftarrow p$};
\draw[thick] (3,-0.25) -- (5,-0.25) node[right,black] {$\leftarrow 1_B$};
\end{tikzpicture}

\noindent
Read in the Schr\"odinger picture, this diagram takes a state $\r$ as input, and produces a state $\Mtilde^*(\r)$ on $C(X)\ten\A$ as output:
   $$\Mtilde^*(\r)(1_B\ten p)=\r\bigl(M(B)[p]\bigr)\;.$$
Straight lines represent commutative C*-algbras, wiggly lines non-commutative ones.
Parallel lines stand for tensor products. 

\section{Weak Convergence}\label{WeakConv}
The main tool of our proof of Holevo's representation theorem will be weak convergence of $\L(\A)$-valued and $\A$-valued measures.

\noindent
By the {\it weak* topology} on the space $\Mctilde$ of continuous functionals $C(X)\to E$ we shall mean the weakest topology that makes
the functional $\Mtilde\mapsto\Mtilde(f)$ continuous for all $f\in C(X)$.
Note that we deviate here slightly from the usual convention, where the weak* topology is defined on the dual of a Banach space,
the space of bounded linear functionals $C(X)\to\CC$.
Here we consider the "$E$-dual" of $C(X)$.
Using the Riesz Representation Theorem we may carry over this topology to the space of $E$-valued measures on $X$,
where probabilists usually call it "the topology of weak convergence".

\begin{definition}[Weak Convergence on Compact $X$]\label{DefWeakConvX}
A sequence $(\mu_n)$ of $E$-valued  measures on $X$ {\it converges weakly} to an $E$-valued measure $\mu$ if for all $f\in C(X)$ we have $\mutilde_n(f)\tto\mutilde(f)$, i.e.:
   $$\li n \int_X f(x)\mu_n(dx)=\int_X f(x)\mu(dx)\;.$$  
\end{definition}

\noindent
We state without proof a simple extension of the sequential Banach-Alaoglu Theorem:

\begin{pro}[Banach-Alaoglu]\label{BanachAlaoglu}
Let $X$ be a compact metric space, and $E$ a finite dimensional normed space.
Then the unit ball in the space of continuous functionals $C(X)\to E$ is sequentially weak* compact.

\noindent
In particular, via the Riesz Representation in Proposition \ref{Riesz},
every bounded sequence $(\mu_n)_{n\in\NN}$ of $E$-valued measures on a compact metric space $X$ has a weakly convergent subsequence.
\end{pro}


\subsection{Addition of outcomes}\label{Addition}
In classical probability infinite sequences of --- for example --- coin tosses are described by sequences of independent identically distributed
0-1-valued random variables, described by a product measure on $\{0,1\}^\NN$.
However, the usual transition to continuous time is {\it not} obtained by considering product measures on $\{0,1\}^\RR$.
This approach would run into problems of measurability.
Instead, as is well-known, the sequences of 0's and 1's from the discrete case are first added up, and only then a proper continuum limit is taken.
This leads to the rich theory of stochastic processes such as Brownian motion, Poisson processes, and more generally, L\'evy processes with stationary independent increments.

\noindent
So in order to make progress on quantum measurement in continuous time we must have {\it addition} of measurement outcomes
to our disposal.
In fact a version of Holevo's theory, where $X$ is an arbitrary Lie group, has been developed\cite{BarHolLup93}.
Here we shall take for $X$ the real numbers.
Simple as this case may seem, we {\it do} meet a difficulty: 
the compactness of $X$, leading to the separability of $C(X)$, was an essential condition for the existence of convergent
subsequences in Proposition \ref{BanachAlaoglu} that we are going to rely on.
But $\RR$ is not compact.
The solution, usually taken in the probability literature, is to manage {\it tightness}.
We may formulate this as follows:
We take for $X$ the one-point compactification $\Rbar:=\RR\cup\{\infty\}$ of $\RR$.
We work on $C(\Rbar)$.
But in practice we are not interested in the outcome $\infty$, so we shall take care that our measures stay away from this point.
We shall have $\mu(\{\infty\})=0$.\ 
Let $C_b(\RR)$ denote the space of all bounded continuous functions $\RR\to\CC$.
In addition to Definition \ref{DefWeakConvX} we shall be using:

\begin{definition}[Weak Convergence on $\RR$]\label{WeakConvR}
A sequence $(\mu_n)$ of measures on $\RR$ {\it converges weakly} to a measure $\mu$ {\it on $\RR$} if for all $f\in C_b(\RR)$ we have
   $$\li n \intR f(x)\mu_n(dx)=\intR f(x)\mu(dx)\;.$$  
\end{definition}

\noindent
{\bf Remark:}
For measures $\mu$ with $\mu(\{\infty\})=0$, weak convergence on $\RR$ is {\it stronger} than weak convergence on $\Rbar$. 
Without explicit mention of the domain ($\RR$ or $\Rbar$), we shall mean weak convergence on $\RR$ as in Defintion \ref{WeakConvR}.

\noindent
In order to bridge the gap between $C(\Rbar)$ and $C_b(\RR)$, the following notion is introduced:

\subsection{Tightness}\label{tightness}
\begin{definition}[Total Variation]\label{DefTotVar}
By the {\it total variation measure} of  an $E$-valued measure $\mu$ on $\RR$ we mean the positive Radon measure $|\mu|$ given by
   $$|\mu|(B):=\sup_f\Re\int_B\inp{f(x)}{\mu(dx)}\;,$$
for all $B\in\B(\RR)$,
where the supremum is to be taken over all continuous functions $f:\RR\to E^*$ with $\norm{f(x)}_{E^*}\le1$ for all $x\in\RR$.
\end{definition}\noindent

\noindent
For the case $E=\L(\A)$, and a completely positive measure $M$, the total variation takes a particularly nice form.
Let $\b$ be the "maximally entangled vector" in $\CC^d\ten\CC^d$:
   $$\b:=\frac1{\sqrt d}\sod i e_i\ten e_i\;,$$
$(e_i)$ being the canonical orthonormal basis of $\CC^d$.
Then the {\it Choi-Jamio\l kowski isomorphism}\cite{Choi75} provides a one-to-one correspondence between linear maps $A:\A\to \A$ and complex functionals 
$\oo_A:\A\ten \A\to\CC$, given by
    $$\oo_A(x\ten y):=\inp{\b}{(Ax\ten y)\b}\;.$$
An element $\a$ in the dual of $\L(\A)$ thus gets mapped to an element $z_\a$ of the second dual of $\A\ten \A$, which is $\A\ten \A$ itself, with the action:
\begin{equation}
\inp \a A=\inp {z_\a}{\oo_A}=\oo_A(z_\a)=\inp\b{(A\ten\id)[z_\a]\b}\;.
\end{equation}
A great advantage of the Choi isomorphism is, that $A$ is completely positive if and only if $\oo_A$ is positive\cite{Choi75} \cite{Buch}.
We obtain a natural norm on $\L(\A)$  by putting
   $$\norm A:=\norm{\oo_A}=\sup\bigl\{|\oo_A(z)|\big|z\in \A\ten \A,\norm z\le1\bigr\}\;.$$
In terms of the density matrix $\r_A$ of $\oo_A$ this is the {\it trace norm} $\tr\sqrt{\r_A^*\r_A}$.

\smallskip\noindent{\bf Warning:}
In Section \ref{NDQuantMeas} we introduced on $\L(\A)$ the structure of a C*-algebra, see also Section \ref{LinAlg}.
This induces another norm, and also a different notion of positivity on $\L(\A)$.
\begin{lem}\label{ThmTotVarM}
Let $M$ be a completely positive $\L(\A)$-valued measure on $\RR$. Then, for all Borel subsets $B$ of $\RR$,
   $$|M|(B)=\frac1d\tr\bigl(M(B)[\one]\bigr)\le\norm{M(B)[\one]}\;.$$
In particular, for an instrument $M$ we have $|M|(\RR)=1$.
\end{lem}

\begin{proof}
\begin{eqnarray*}
|M|(B)&=&\sup_f\Re\int_B\inp{f(x)}{M(dx)}=\sup_f\Re\int_B\inp{\b}{\bigl(M(dx)\ten\id\big)[f(x)]\b}\;,
\end{eqnarray*}
where $f$ runs over the continuous functions from $\RR$ to the unit ball in $\L(\A)^*\cong\A\ten \A$.
Since $M$ is completely positive, the functional $z\mapsto\inp\b{\bigl(M(C)\ten\id\bigr)[z]\b}$ is positive for each $C\in\B(\RR)$,
hence  takes its maximum over the unit ball in $\one\ten\one$.
Therefore we may take $f(x)=\one\ten\one$:
\begin{eqnarray*}
|M|(B)&=&\inp\b{\bigl(M(B)\ten\id\bigr)[\one\ten\one]\b}\\
          &=&\frac1d\sod{i,j}\Inp{e_i\ten e_i}{\bigl(M(B)[\one]\ten\one\bigr) e_j\ten e_j}=\frac1d\tr\bigl(M(B)[\one]\bigr)\;,
\end{eqnarray*}
and, of course, the average of the (positive) eigenvalues of $M(B)[\one]$ is less than their maximum $\norm{M(B)[\one]}$.
Finally, if $M$ is an instrument,
   $$|M|(\RR)=\frac 1d\tr\bigl(M(\RR)[\one]\bigr)=\frac 1d\tr\one=1\;.$$
\end{proof}

\begin{definition}\label{DefTight}
A bounded sequence $(\mu_n)_{n\in\NN}$ of $E$-valued measures on $\RR$ is called {\it tight} if for all $\e>0$ there is a $K>0$ with the property that
   $$\forall_{n\in\NN}:\quad |\mu_n|\bigl(\RR\setminus[-K,K]\bigr)\le\e\;.$$
\end{definition}

\noindent
Every measure $\mu$ on $\RR$ will be considered as a measure on $\Rbar$ by putting $\mu(\{\infty\})$ $:=0$.

\begin{lem}\label{LemTussen}
Let $(\mu_n)_{n\in\NN}$ be a tight sequence of $\L(\A)$-valued measures on $\RR$ which converges weakly on $\Rbar$ to a measure $\mu$ on $\Rbar$.
Then $\mu(\{\infty\})=0$ and $\mu_n$ converges weakly to $\mu$ on $\RR$.
\end{lem}

\noindent
Let $\tuple e m$ be a basis of $E$, and $\tuple {e^*} m$ the dual basis of $E^*$.

\begin{proof}
We may assume that $|\mu_n|(\RR), |\mu(\Rbar)|\le1$.
Note that, by Definition \ref{DefTotVar} and weak convergence,
\begin{eqnarray*}
|\mu|(\Rbar\setminus[-K,K])&=&\sup_f\Re\int_{\Rbar\setminus[-K,K]} \inp{f(x)}{\mu(dx)}\\
                                           &=&\sup_f\Re\li n\int_{\Rbar\setminus[-K,K]} \inp{f(x)}{\mu_n(dx)}\;,
\end{eqnarray*}
where the supremum is over $f\in C(\Rbar,E^*)$ with $\norm f\le1$.
Let $\e>0$. By tightness, we can choose $K$ so that for all $n\in\NN$ the integrals on the right are all less than $\e$.
Then so is their limit, and also the supremum thereof.
By the outer regularity of $|\mu|$:
   $$|\mu|(\{\infty\})=\inf_{K>0}|\mu|\bigl(\Rbar\setminus[-K,K]\bigr)=0\;.$$

\noindent
In order to show weak convergence on $\RR$ we must extend the following convergence, which holds for all $f\in C(\Rbar)$
since $\mu_n(\{\infty\})=\mu(\{\infty\})=0$:
\begin{equation}\label{EqWeakConv}
\li n\int_{\Rbar} f(x)\mu_n(dx)=\int_{\Rbar} f(x)\mu(dx)
\end{equation}
to all $f\in C_b(\RR)$.
Without loss of generality we may assume that $\norm f\le\half$.
Let $\e>0$, and let $K>0$ be such that $|\mu_n|(\RR\setminus[-K,K])<\e$ for all $n\in\NN$, and also $|\mu|(\RR\setminus[-K,K])<\e$. 
Let $h\in C(\Rbar)$ be such that $\norm h\le\half$, $h(x)=f(x)$ for $|x|\le K$, and $h(\infty)=\lim_{x\to\pm\infty}h(x)=0$;
then let $g_{j,\psi}:\RR\to E^*$ with $j=1,\ldots,m$ and $\psi\in[0,2\pi]$ be given by $g_{j,\psi}(x):=e^{i\psi}\bigl(f(x)-h(x)\bigr)e_j^*$.
Then $\norm{g_{j,\psi}}\le1$ and $g_{j,\psi}=0$ on $[-K,K]$, hence  by Definition \ref{DefTotVar} we have, since $(\mu_n)$ is tight
and $g_{j,\psi}\in C_b(\RR,E^*)$,
   $$\e>\Re\intR\inp{g_{j,\psi}(x)}{\mu_n(dx)}=\Re e^{i\psi}\intR\bigl(f(x)-h(x)\bigr)\inp{e_j^*}{\mu_n(dx)}\;.$$
Since this holds for all $\psi$, we have for all $j$, putting $\inp{e_j^*}{\mu_n(B)}=:\mu_n^j(B)$,
and the same for $\mu$:
   $$\left|\intR\bigl(f(x)-h(x)\bigr)\mu_n^j(dx)\right|<\e\;.$$
Applying (\ref{EqWeakConv}) to $h\in C(\Rbar)$, there exists $N\in\NN$ such that for $n\ge N$ we have
   $$\left|\intR h(x)\mu_n^j(dx)-\intR h(x)\mu^j(dx)\right|<\e\;.$$
Thus we obtain for $n\ge N$:
\begin{eqnarray*}
\left|\int fd\mu_n^j-\int fd\mu^j\right|
          &\le&\left|\int(f-h)d\mu_n^j\right|+\left|\int hd\mu_n^j-\int hd\mu^j\right|\\
          &+&\left|\int(h-f)d\mu^j\right|\le3\e\;,
\end{eqnarray*}
proving (\ref{EqWeakConv}) for all $f\in C_b(\RR)$. So we have $\mu_n\tto\mu$ weakly on $\RR$.
\end{proof}


\section{Convolution of Instruments and L\' evy's Continuity Theorem}\label{LevyCont}
Performing one instrument $M$, and subsequently another instrument $N$, and adding up their outcomes, yields a new instrument, the {\it convolution} of $M$ and $N$.

\begin{definition}\label{DefConv}
The {\it convolution} $M*N$ of the instruments $M$ and $N$ on $\RR$ is defined by
\begin{equation*}
M*N(B)[p];=\int_\RR\int_\RR1_B(x+y)M(dx)\bigl[N(dy)[p]\bigr]\;.
\end{equation*}
for all Borel subsets $B$ of $\RR$ and projections $p\in\A$.
There is an immediate extension to finite $\L(\A)$-valued measures.
\end{definition}

\subsection{Side issue: addition as an operation on probability spaces}
Addition of real numbers leads to a positive unital map $A:C(\Rbar)\to C(\Rbar)\ten C(\Rbar)$ given by
   $$\bigl(Af\bigr)(x,y):=f(x+y)\;.$$
(We take $x+y=\infty$ as soon as $x$ or $y$ equals $\infty$.)
We may draw a diagram for this situation, to be read from right to left (in the "Heisenberg picture"):

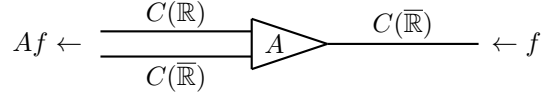
\begin{tikzpicture}
\draw[thick] (5,-0.333) -- (7,-0.333); 
\draw(6,-0.6) node {$C(\Rbar)$};
\draw(6,0.25) node {$C(\Rbar)$};
\draw[thick] (5,0) -- (7,0);
\draw[thick] (7,-0.5) -- (7,0.166);
\draw[thick] (7,-0.5) -- (8,-0.166);
\draw[thick] (7,0.166) -- (8,-0.166);
\draw (7.3,-0.166) node {$A$};
\draw[thick] (8,-0.166) -- (10,-0.166);
\draw(9,0.1) node  {$C(\Rbar)$};
\draw (10.5,-0.166) node {$\leftarrow f$};
\draw(4.3,-0.166) node {$Af\leftarrow$};
\end{tikzpicture}

\noindent
When read from left to right (in the "Schr\"odinger picture") this yields a map $A^*:C(\Rbar)^*\ten C(\Rbar)^*\to C(\Rbar)^*$ given by the diagram

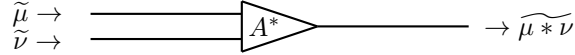
\begin{tikzpicture}
\draw[thick] (5,-0.333) -- (7,-0.333); 
\draw[thick] (5,0) -- (7,0);
\draw[thick] (7,-0.5) -- (7,0.166);
\draw[thick] (7,-0.5) -- (8,-0.166);
\draw[thick] (7,0.166) -- (8,-0.166);
\draw (7.3,-0.166) node {$A^*$};
\draw[thick] (8,-0.166) -- (10,-0.166);
\draw (10.8,-0.166) node {$\rightarrow\widetilde{\mu*\nu}$};
\draw(4.3,-0.333) node {$\nutilde\rightarrow$};
\draw(4.3,0) node {$\mutilde\rightarrow$};
\end{tikzpicture}

\noindent
So $A^*$ is convolution of measures.
In this vein, as the reader may check, Definition \ref{DefConv} can be represented by the diagram

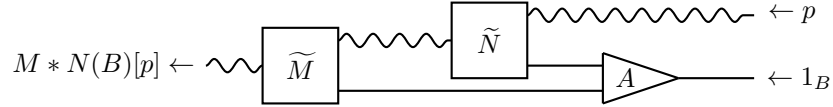
\begin{tikzpicture}
\draw (-0.6,0) node {$M*N(B)[p]\leftarrow$};
\draw[thick,decorate,decoration={snake}] (0.75,0) -- (1.5,0);
\draw[thick] (1.5,-0.5) rectangle (2.5,0.5);\draw (2,0) node {$\Mtilde$};
\draw[thick,decorate,decoration={snake}] (2.5,.333) -- (4,.333);
\draw[thick] (2.5,-0.333) -- (6,-0.333); 
\draw[thick] (4,-0.166) rectangle (5,0.833); \draw (4.5,0.333) node {$\Ntilde$}; 
\draw[thick,decorate,decoration={snake}] (5,0.666) -- (8,.666);
\draw[thick] (5,0) -- (6,0);
\draw[thick] (6,-0.5) -- (6,0.166);
\draw[thick] (6,-0.5) -- (7,-0.166);
\draw[thick] (6,0.166) -- (7,-0.166);
\draw (6.3,-0.166) node {$A$};
\draw[thick] (7,-0.166) -- (8,-0.166);
\draw (8.6,-0.166) node {$\leftarrow 1_B$};
\draw (8.5,0.666) node {$\leftarrow p$};
\end{tikzpicture}

\noindent
When read in the Schr\"odinger picture, i.e., from left to right, this diagram depicts quite neatly what is meant by convolution of instruments.

\subsection{Characteristic functions of instruments}

\begin{definition}
The {\it characteristic function} of a finite $E$-valued measure $\mu$ is the continuous function $\muhat:\RR\to E$ given by
   $$\muhat(\no):=\intR e^{i\no x}\mu(dx)\;.$$
\end{definition}

\begin{lem}\label{LemProdCharFns}
For finite $\L(\A)$-valued measures $M$ and $N$ on $\RR$ we have
\begin{equation*}
\widehat{M*N}(\no)[a]=\Mhat(\no)\circ\Nhat(\no)[a]\;.
\end{equation*}
\end{lem}

\begin{proof}
From Definition \ref{DefConv} we obtain
\begin{eqnarray*}
\widehat{M*N}(\oo)&=&\intR\intR e^{i\no(x+y)}M(dx)\circ N(dy)\\
                       &=&\left(\intR e^{i\no x}M(dx)\right)\circ\left(\intR e^{i\no y}N(dy)\right)=\Mhat(\no)\circ\Nhat(\no)\;.
\end{eqnarray*}
\end{proof}

\noindent
Here is a technical result for later use: 
\begin{pro}\label{PropWeakToChar}
If a tight sequence $(\mu_n)$ converges weakly to $\mu$ on $\RR$, then
   $$\li n\muhat_n(\oo)=\muhat(\oo)$$
uniformly on compact subsets of $\RR$.
\end{pro}

\begin{proof}
(Sketch)
If $\displaystyle\wlim_{n\to\infty}\mu_n=\mu$, then we have for all $\oo\in\RR$:
   $$\muhat_n(\oo)=\intR e^{i\oo x}\mu_n(dx)\tto\intR e^{i\oo x}\mu(dx)=\muhat(\oo)\;.$$
The tightness of $(\mu_n)$ leads to uniform equicontinuity of $(\muhat_n)$,
which in turn yields locally uniform convergence. This is equivalent to uniform convergence on compact sets. See the literature\cite{Shir96} for a detailed proof.
\end{proof}

\subsection{L\' evy's Continuity Theorem}
L\'evy's Continuity Theorem says, broadly speaking, that for positive measures, convergence of characteristic functions corresponds to weak convergence of
measures. We have this also in our case of completely positive measures.
For our purpose, however, we need more detailed information:

\begin{thm}[Extended Version of L\'evy's Continuity Theorem]\label{ThmLevCont}
Let $(\mu_n)_{n\in\NN}$ be a sequence of finite $E$-valued measures on $\RR$, such that for some function $\ph:\RR\to E$:
    $$\li n \muhat_n(\oo)=\ph(\oo)\;.$$
Then we have the implications (a)$\implies$(b)$\implies$(c):
\begin{itemize}
\item[(a)]
$\displaystyle(\mu_n)_{n\in\NN}$ is tight;
\item[(b)]
$\mu_n$ converges weakly on $\RR$ to a finite measure $\mu$ on $\RR$, and $\ph=\muhat$;
\item[(c)]
$\ph$ is continuous at 0.
\end{itemize}
Moreover, in the case that $E=\L(\A)$ and $\mu_n$ is completely positive, we also have (c)$\implies$(a), so that (a), (b), and (c) are equivalent.
\end{thm}

\begin{proof}
(a)$\implies$(b):
By the sequential Banach-Alaoglu Theorem, Proposition \ref{BanachAlaoglu}, some subsequence $(\mu_{n_k})_{k\in\NN}$, converges weakly on $\Rbar$ to a measure $\mu$ on $\Rbar$.
By Lemma \ref{LemTussen}, $\mu$ is actually a measure on $\RR$, since $\mu(\{\infty\})=0$, and $\mu_{n_k}\tto\mu$ weakly  on $\RR$ (Definition \ref{WeakConvR}).
By choosing the test function $f(x)=e^{i\oo x}$ we obtain that
   $$\ph(\oo)=\li k\muhat_{n_k}(\oo)=\muhat(\oo)\;.$$
Since $\muhat$ determines $\mu$, all convergent subsequences of $(\mu_n)$ converge to the same measure $\mu$.
Suppose that we do not have $\wlim_{n\to\infty}\mu_n=\mu$ on $\RR$.
Then there exist $f\in C_b(\RR)$, $\e>0$, and an infinite subfamily $\set{\mu_n}{n\in F}$, $F\subset\NN$ with
\begin{equation}\label{EqDefactors}
\forall_{n\in F}:\quad\left|\intR f(x)\mu_n(dx)-\intR f(x)\mu(dx)\right|\ge\e\;.
\end{equation}
But by  Proposition \ref{BanachAlaoglu} and Lemma \ref{LemTussen}
this subfamily again contains a convergent subsequence, which can only converge to $\mu$, in contradiction to (\ref{EqDefactors}).
So $\mu_n\tto\mu$ weakly on $\RR$.

\noindent(b)$\implies$(c):
The characteristic function of a finite measure is continuous.

\noindent
(c)$\implies$(a):
Suppose that $E=\L(\A)$ and the measures $\mu_n$ (now called $M_n$), are completely positive,
$\Mhat_n$ converging pointwise to a function $\ph$, continuous at 0.
Let $\e>0$.
Since $\ph$ is continuous at 0, by a $3\e$-argument there exist $\d>0$ and $N\in\NN$ such that for all $n\ge N$:
   $$|\oo|<\d\implies\Norm{\Mhat_n(0)[\one]-\Mhat_n(\oo)[\one]}<\half\e\;.$$
Let $K:=\frac2\d$.
Then $|\d x|>2$ for $x\in[-K,K]^c$, and for all $x$: $\frac{\sin x}x\le\min(1,\frac1{|x|})$ (where we read $\frac{\sin x}x=1$ for $x=0$).
Hence we have for all $x\in\RR$:
\begin{eqnarray*}
\half1_{[-K,K]^c}(x)&\le&\left(1-\frac1{|\d x|}\right)\cdot 1_{[\frac2\d,\infty)}\bigl(|x|\bigr)
             \le1-\frac{\sin\d x}{\d x}\\
             &=&\frac1\d\int_0^\d(1-\cos\oo x)d\oo=\frac1{2\d}\int_{-\d}^\d\bigl(1-e^{i\oo x}\bigr)d\oo\;.
            \end{eqnarray*}
As $M_n$ is a finite completely positive  measure, we may integrate both sides with respect to $M_n$ and multiply by 2, to obtain, in the completely positive ordering,
\begin{eqnarray*}
M_n\bigl([-K,K]^c\bigr) &\le&\intR\left(\frac1{\d}\int_{-\d}^\d\left(1-e^{i\oo x}\right)d\oo\right)M_n(dx)\\
                                    &=&\frac1{\d}\int_{-\d}^\d\bigl(\Mhat_n(0)-\Mhat_n(\oo)\bigr)d\oo\;.
             \end{eqnarray*}
By Lemma \ref{ThmTotVarM}, for $n\ge N$,
   $$|M_n|([-K,K]^c)\le\norm{M_n([-K,K]^c)[\one]}\le2\sup_{\oo\in[-\d,\d]}\Norm{\Mhat_n(0)[\one]-\Mhat_n(\oo)[\one]}\le\e\;.$$
By making $K$ larger we can ensure that the finite measures $M_1,\ldots M_{N-1}$ satisfy this inequality as well.
We conclude that $(M_n)_{n\in\NN}$ is tight.
\end{proof}


\section{Convolution Semigroups of Instruments}\label{ConvSem}
Observation in continuous time is described by semigroups of instruments.

\begin{definition}
By a {\it convolution semigroup of instruments} we mean a family $(M_t)_{t\ge0}$ of instruments on $\RR$, satisfying $M_0=\d_0\cdot\id_\A$, and
for all $s,t\ge0$:
   $$M_{s+t}=M_s*M_t\;.$$
Such a semigroup is called {\it weakly continuous} if  for all $t \ge0$.
   $$\wlim_{s\to t}M_s=M_t\;.$$
\end{definition}
\noindent
Note that for the latter condition, by writing $s=u+t$, and bracketing out an operator $M_t$,
it suffices that $\displaystyle\wlim_{u\downarrow0}M_u=\d_0\cdot\id_\A$.

\noindent{\bf Discussion.}
Every instrument $M_t$ that lies in such a semigroup is automatically {\it infinitely divisible}:
for every $n\in\NN$ it is the $n$-th convolution power of some distribution $M_{t/n}$.
In classical probability a converse statement holds true: given an infinitely divisible distribution $\mu$,
there is a unique weakly continuous convolution semigroup $(\mu_t)$ such that $\mu=\mu_1$.
We do not know, however, what would be a suitable generalization of this in our framework.
For this reason, we take as our starting point for continuous observation the whole semigroup,
not a single infinitely divisible instrument. 

\begin{pro}[Characteristic Exponent]\label{ThmCharExp}
Let $(M_t)$ be a family of instruments on $\RR$.
Then the following are equivalent
\begin{itemize}
\item[(a)]
$(M_t)$ is a weakly continuous convolution semigroup;
\item[(b)]
there exists a continuous function $L:\RR\to\L(\A)$ such  that for all $t\ge0$ and $\oo\in\RR$
\begin{equation}\label{EqCharExp}
\Mhat_t(\oo)=e^{tL(\oo)}\;.
\end{equation}
\end{itemize}
\end{pro}

\begin{definition}
If these equivalent conditions hold, then $L$ is called the {\it characteristic exponent} or {\it generator} of the semigroup $(M_t)$.
\end{definition}

\noindent
This is the central notion of this paper.

\begin{proof}
(a)$\implies$(b):
By Lemma \ref{LemProdCharFns} the semigroup property $M_{s+t}=M_s*M_t$ implies that for all $\oo\in\RR$ we have, for all $s,t\ge0$,
   $$\Mhat_{s+t}(\oo)=\Mhat_s(\oo)\circ\Mhat_t(\oo)\;.$$
Moreover, by weak continuity,
   $$\lim_{s\to t}\Mhat_s(\oo)=\lim_{s\to t}\intR e^{i\oo x}M_s(dx)=\intR e^{i\oo x}M_t(dx)=\Mhat_t(\oo)\;.$$
So $t\mapsto \Mhat_t(\oo)$ is a continuous semigroup of operators on the finite dimensional space $\A$ for each $\oo\in\RR$.
Such a semigroup is automatically differentiable\cite{DaviesSemigroups}.
Let $L(\oo)$ be its derivative at 0. Then for all $t\ge0$:
\begin{eqnarray*}
\frac d{dt}\Mhat_t(\oo)&=&\lim_{s\downarrow0}\frac1s\bigl(\Mhat_{t+s}(\oo)-\Mhat_t(\oo)\bigr)\\
                                   &=&\lim_{s\downarrow0}\frac1s \Mhat_t\circ\bigl(\Mhat_s(\oo)-\id_\A\bigr)=\Mhat_t(\oo)\circ L(\oo)\;.
         \end{eqnarray*}
The solution of this differential equation is (\ref{EqCharExp}).

\noindent
We claim that $L(\oo)$ is continuous in $\oo$.
Let $\oo_0\ge0$.
First note that $(M_t)_{t\ge0}$ is  tight by Theorem \ref{ThmLevCont}, (b)$\implies$(a).
Therefore, by Proposition \ref{PropWeakToChar}, the convergence 
of $\Mhat_t(\no)$ to $\Mhat_0(\oo)=\id_{\A}$, holds uniformly on $[-\no_0,\no_0]$:
  $$\lid t\max_{|\no|\le\no_0}\norm{\Mhat_t(\no)-\id_{\A}}=0\;.$$
Denote, for $A\in\L(\A)$ and $\e>0$, the unit ball with radius $\e$ around $A$ by $\B_\e(A)$.
For $\e$ small enough, the exponential map restricted to $\B_\e(0)$ is invertible on its range; call the inverse $\log_\e$.
Let $\d>0$ be such that $\B_\d(\id_{\A})$ is contained in this range $\exp\bigl(\B_\e(0)\bigr)$.
Let $t_0$ be such that $\norm{\Mhat_t(\no)-\id_{\A}}<\d$ for all $\no\in[-\no_0,\no_0]$ and $t<t_0$.
It follows from (\ref{EqCharExp}) that, for $0<t<t_0$ and $-\no_0\le\no\le\no_0$:
\begin{equation}\label{EqLogEps}
L(\no)=\frac 1t\log_\e\bigl(\Mhat_t(\no)\bigr).
\end{equation}
Since $\log_\e$ is a smooth map, and $\oo\mapsto\Mhat_t(\oo)$ is continuous,
$L$ is continuous on $[-\no_0,\no_0]$. As this holds for all $\no_0>0$, it is continuous on the whole of $\RR$.

(b)$\implies$(a):
The condition (\ref{EqCharExp}) implies that for all $s,t\ge$ and $\oo\in\RR$
   $$\Mhat_s(\oo)\circ\Mhat_t(\oo)=e^{sL(\oo)}\circ e^{tL(\oo)}=e^{(s+t)L(\oo)}=\Mhat_{s+t}(\oo)\;.$$
By Lemma \ref{LemProdCharFns} this implies that
   $$\widehat{\bigl(M_s*M_t\bigr)}(\oo)=\Mhat_{s+t}(\oo)\;.$$

\noindent
Since an instrument is determined by its characteristic function, it follows that $M_s*M_t=M_{s+t}$.
Moreover, we have for all $\oo\in\RR$:
   $$\lim_{t\downarrow0}\Mhat_t(\oo)=\lim_{t\downarrow0}e^{tL(\oo)}=\id_\A=\Mhat_0(\oo)\;,$$
which is certainly continuous at $\oo=0$. By Theorem \ref{ThmLevCont} it follows that $M_t\tto M_0$ weakly.
\end{proof}


\subsection{Example: the basic Davies generator}\label{BasicDavies}
For $a,b\in\A=M_d$, we denote the {\it anticommutator} $ab+ba$ by $\{a,b\}$.
\begin{pro}[Davies Generator]\label{PropBasicDavies}
Let for $x\in\RR$ and $a\in\A$ the function $L_{x,a}:\RR\to\L(\A)$ be given by
\begin{equation*}
L_{x,a}(\no)[z]:=e^{i\no x}a^*za-\half\{a^*a,z\}\;.
\end{equation*}
Then $L_{x,a}$ is the generator of a weakly continuous convolution semigroup of instruments.
\end{pro}

\begin{proof}Write
\begin{equation}
J[z]:=a^*za\quad\hbox{and}\quad K[z]:=-\half\{a^*a,z\}.
\end{equation}
Then $J$ is completely positive, as well as
    $$e^{tK}:z\mapsto e^{-\frac t 2 a^*a}ze^{-\frac t 2 a^*a}\;.$$
Define completely positive maps $P_n(t)$ ($t\ge0$, $n\in\NN$) by:  $P_0(t):=e^{tK}$ and for $n\ge1$:
$$P_n(t):=\int_0^t\int_0^{t_n}\cdots\int_0^{t_2}e^{t_1K}\circ J\circ e^{(t_2-t_1)K}\circ\cdots\circ J\circ e^{(t-t_n)K}dt_1dt_2\ldots dt_n\;.$$
(Note that, for $\A=\CC$, and, say, $a^*a=:\lambda$, this is the Poisson probability $P_n(t)=\frac{t^n}{n!}\cdot \lambda^n e^{-\lambda t}$.) 
Define $P_n^\no(t)$ for $\oo\in\RR$ by replacing $J$ with $e^{i\no x}J$ in the above expression for $P_n(t)$.
Let $M_t$ be the completely positive measure on $\RR$ given by
   $$M_t:=\szi n \d_{nx}\cdot P_n(t)\;.$$
Since the factor $J$ occurs $n$ times in $P_n(t)$, the characteristic function can be written:
   $$\Mhat_t(\no)=\szi n e^{in\no x}P_n(t)=\szi n P_n^\no(t)\;.$$
We have $\ddt P_0^\no(t)=P_0^\no(t)\circ K$, and for $n\ge0$,
by differentiation of $P_n^\no(t)$ first "with respect to the upper limit", and then  "under the integral", we find
   $$\ddt P_n^\no(t)=P_{n-1}^\no(t)\circ\bigl(e^{i\no x}J\bigr)+P_n^\no(t)\circ K\;.$$
Therefore
\begin{eqnarray*}
\ddt\Mhat_t(\no)&=&P_0^\no(t)\circ K+\soi n \Bigl(P_{n-1}^\no(t)\circ\bigl(e^{i\no x}J\bigr)+P_n^\no(t)\circ K\Bigr)\\
                        &=&\Mhat_t(\no)\circ\bigl(e^{i\no x}J+K\bigr)=\Mhat_t(\no)\circ L_{x,a}(\no)\;.
\end{eqnarray*}
Solving this differential equation with initial condition $\Mhat_0(\no)=\widehat{\d_0}(\oo)\cdot\id_\A=\id_\A$ gives $\Mhat_t(\no)=e^{tL_{x,a}(\no)}$.
By Proposition \ref{ThmCharExp},  $(M_t)$ is a weakly continuous convolution semigroup.
Moreover, since $L_{x,a}(0)[\one]=a^*\one a-\half\{a^*a,\one\}=0$,
we have $M_t(\RR)[\one]=\Mhat_t(0)[\one]=e^{tL_{x,a}(0)}[\one]=\one$, so the measure $M_t$ is an instrument.
\end{proof}

\noindent
\subsection{Discussion}
The basic Davies generator is the central building block of all characteristic exponents.
It describes a point process of observations ("clicks"), which accompanies the quantum evolution.
A click might be the detection of a particle emitted by the quantum system.
When the initial state is $\r_0$, then the distribution of the number of particles counted up to time $t$ is $\r_0\bigl(M_t(\one)\bigr)$.
A "quantum jump" $J$, where $\psi$ jumps to $a\psi$, occurs at every click,
or maybe it should be said that the detection of a particle leaving the system,
suddenly changes the information which the observer possesses about the system.
A typical property of quantum mechanics is, that the two situations cannot be clearly distingished.
Between clicks the quantum system moves according to the smooth evolution $e^{tK}$.
The latter does not only describe the free evolution of the system, but  also the effect of conditioning on the non-occurrence of clicks.
Hence $e^{tK}$ typically contracts to 0, indicating that a pause of length $t$ between clicks is increasingly rare.
This is to be compared with the waiting time distribution function $F(t)=1-e^{-\lambda t}$ of a Poisson process.

\subsection{Application: resonance fluorescence}\label{ResFluor}
Let $\A=M_2$, $x=1$, and let $a,h\in M_2$ be given by
   $$a:=\sqrt\g\begin{pmatrix}0&1\\0&0\end{pmatrix},\quad h:=\frac i 2\begin{pmatrix}0&-\Oo\\ \Oo&0\end{pmatrix}\;.$$
Let $L(\no)[z]:=L_{x,a}(\no)[z]-i[h,z]$, the basic Davies generator, enriched with a Hamiltonian term.
We may absorb the Hamiltonian term into the operator $K$ in the above, since for $K[z]=-\half\{a^*a,z\}-i[h,z]$ the operator
   $$e^{tK}:z\mapsto e^{-t(\frac12 a^*a+ih)}ze^{-t(\frac12 a^*a-ih)}$$
is still completely positive.
The semigroup $(M_t)$ which we then obtain describes the counting of photons emitted from a two-level atom,
which is irradiated by a laser, tuned to the resonance frequency.
This frequency does not occur in the expression, but there is a rotation with the {\it Rabi  frequency} $\Oo$,
that is proportional to the {\it amplitude} of the (classical) laser field.
The positive constant $\g$ denotes the rate of spontaneous emission of photons, accompanied by a jump of the atom from its upper level $\twosvector01$
to its lower level $\twosvector10$.
(We adopt the qubit convention, according to which $|0\rangle$ denotes
$\twosvector10$, and $|1\rangle$ denotes $\twosvector01$).

\noindent
The photons are emitted according to a renewal process with waiting time distribution function\cite{Buch}
   $$F(t)=1-\Norms{e^{-t(\frac 12 a^*a-ih)}\twovector10}\;,$$
whose density function is given by\cite{Car99}\cite{Buch}
   $$f(t)=F'(t)=\frac{4\g\Oo^2}{\g^2-4\Oo^2}e^{-\frac12\g t}\sinh^2\left(\frac t4\sqrt{\g^2-4\Oo^2}\right)\;.$$
The phenomenon that $f(0)=0$ is known as "antibunching" of photons:
after emission of one photon the atom has to be reloaded before the next photon can be sent out.
In the 1970's this was experimentally confirmed, and taken as a sign that light does not follow a classical description.\cite{KimbleMandel}


\section{A Quantum de Finetti Theorem}\label{InfDivInstQuantDeFinetti}
It was de Finetti who first realized that, in probability theory,
the characteristic exponents of infinitely divisible distributions on the real line form a closed cone,
generated by the exponents $\no\mapsto e^{i\no x}-1$ of the Poisson distributions of different jump sizes $x\in\RR$.
The following theorem is an extension of this observation to instruments. 

\begin{thm}\label{ThmDeFinetti}
Let $L:\RR\to\L(\A)$. Then the following are equivalent:
\begin{itemize}
\item[(a)]
$L$ is the characteristic exponent of a weakly continuous  convolution semigroup $(M_t)_{t\ge0}$ of instruments on $\RR$;
\item[(b)]
$L$ lies in the pointwise closure $\Cbar$, inside the space of continuous functions, of the cone $\C$ generated by the basic Davies generators
   $$L_{x,a}(\no)[z]:=e^{i\no x}a^*za-\half\{a^*a,z\}\;,\quad x\in\RR, a\in \A\;.$$
\end{itemize}
Moreover, if $L\in\Cbar$ is a pointwise limit of functions $L_n\in\C$, then the convergence $L_n\tto L$ is uniform on compact subsets of $\RR$.
\end{thm}

\noindent
Notation: let $\G$ denote the class of characteristic exponents ("generators") in (a).
Then the theorem can be briefly expressed as
\begin{equation}\label{EqDeFinetti}
\G=\Cbar\;.
\end{equation}

\begin{proof}
From (a) to (b), i.e.: $\G\subset\Cbar$:
Let $L$ be a characteristic exponent.
Then $L$ is continuous by Proposition \ref{ThmCharExp}. We shall approximate $L$ by functions in $\C$.
Let us write $L(\no)=\li n \Ltilde_n(\no)$, where
\begin{eqnarray*}
\Ltilde_n(\no)[z]&=&n\left(e^{\frac1n L(\no)}[z]-z\right)=n\left(\intR e^{i\no x}M_{1/n}(dx)[z]-z\right)\\
          &=&n\intR \biggl(e^{i\no x}M_{1/n}(dx)[z]-\half\{M_{1/n}(dx)[\one],z\}\biggr)\;,
\end{eqnarray*}
since $M_{1/n}(\RR)[\one]=\one$ and $\half\{\one,z\}=z$.
In their turn the finite completely positive measures $nM_{1/n}$ can be approximated weakly by (finite completely positive) measures $J_n$ of he form
\begin{equation}\label{InfDivInstEqJDiscr}
J_n(B)[z]=\sum_{j=1}^{k_n}\d_{x^n_j}(B)\cdot(a_j^n)^*z a_j^n\;.
\end{equation}
for some sequence of integers $k_n$, and sequences of $k_n$-tuples $\tuple {x^n}{k_n}\in\RR$ of jump sizes,
and $\tuple {a^n}{k_n}\in\A$ of jump matrices.

\noindent
We thus obtain that $L(\no)=\li n L_n(\no)$ with
\begin{equation}\label{InfDivInstEqLfromJ}
L_n(\no)[z]=\intR e^{i\no x}J_n(dx)[z]-\half\{J_n(dx)[\one],z\}\;,
\end{equation}
so that $L$ is approximated by sums of basic Davies generators:
\begin{equation}\label{InfDivInstLApprox}
L(\no)[z]=\li n \sum_{j=1}^{k_n}e^{i\no x_j^{n}}(a_j^n)^*za_j^n-\half\{(a_j^n)^*a_j^n,z\}\;.
\end{equation}
We conclude that $\G\subset\Cbar$.

\noindent
From (b) to (a):
We must show that $\Cbar\subset\G$.
\noindent
First note that $L_{x,a}\in\G$ by Proposition \ref{PropBasicDavies}.
It then suffices to show that $\G$ is a cone, closed in the pointwise topology inside the space of all continuous functions $\RR\to\L(\A)$.
Indeed, if $L$ and $L'$ are in $\G$, say $e^{tL}=M_t$ and $e^{tL'}=M'_t$,
then, by the Trotter-Kato formula
\begin{eqnarray*}
\Phi_t(\oo):=e^{t\bigl(L(\no)+L'(\no)\bigr)}&=&\li n e^{\frac t n L(\no)}\circ e^{\frac t n L'(\no)}\circ\cdots\circ e^{\frac t n L(\no)}\circ e^{\frac t n L'(\no)}\\
               &=&\li n \bigl(M_{\frac t n}*M'_{\frac t n}*\cdots*M_{\frac t n}*M'_{\frac t n}\bigr)\widehat{\phantom{A}}(\no)\\
                &=:&\li n\Nhat_t^n(\oo)\;,
\end{eqnarray*}
So $\Phi_t$ is a pointwise limit of characteristic functions of instruments $N^n_t$.
Since $L$ and $L'$ are continuous, so is $\Phi_t$.
By L\'evy's Continuity Theorem \ref{ThmLevCont}, implication (c)$\implies$(b), which holds in the completely positive case,
for each $t\ge0$ there exists an instrument $N_t=\wlim_{n\to\infty}N_t^n$, in particular $\Nhat_t(\oo)=\li n\Nhat_t^n(\oo)=\Phi_t(\oo)=e^{t(L+L')(\oo)}$.
By Proposition \ref{ThmCharExp}, again since $L+L'$ is continuous, $(N_t)$ is a weakly continuous semigroup of instruments, i.e., $L+L'\in\G$.
So $\G$ is a cone.
Finally, if some sequence $L_1,L_2,\ldots$ in $\G$ converges pointwise to a continuous function $L$, then $\displaystyle\li n e^{tL_n(\no)}=e^{tL(\no)}$,
and, again by Theorem \ref{ThmLevCont} and Proposition \ref{ThmCharExp}, we have $L\in\G$. We conclude that $\Cbar\subset\G$.

\noindent
Moreover, by L\'evy's Continuity Theorem \ref{ThmLevCont}, $\bigl(e^{tL_n}\bigr)$ is tight, and $e^{tL_n}\tto e^{tL}$ weakly.
By Proposition \ref{PropWeakToChar},  $e^{tL_n}\tto e^{tL}$ uniformly on compact sets,
and by the argument in (\ref{EqLogEps}), involving $\log_\e$, this property is transferred to the convergence $L_n\tto L$.
\end{proof}

 \section{The path space of the "needle"}\label{Needle}
The semigroup $(M_t)$ of instruments on $\RR$ determines a stochastic process $(X_t)$, $t\in[0,T]$ on $\RR$
via the family of $\L(\A)$-valued "transition kernels"
    $$K_t(x,B):=M_t(B-x)\;,$$
where $x\in\RR$, $B\in\B(\RR)$.
We have chosen a fixed time $T\ge0$ up to which the observations are done.

\noindent
A minimal probability space for this process is $(\Oo_T,\F_T)$,
where $\Oo_T=\RR^{[0,T]}$, and $\F_T$ is the product $\s$-algebra.
The process $(X_t)_{0\le t\le T}$ is given by $X_t:\Oo_T\to\RR:\xi\mapsto\xi(t)$.
As a measure on this space we take an instrument $\M_T:\F_T\to\CP(\A)$,
determined, via the Kolmogorov-Daniell construction, by the consistent family of cylinder measures,
in obvious probabilistic notation given by, for $0< t_1< t_2,\ldots,t_{n-1}<t_n=T$,
\begin{eqnarray}\label{EqDefMarkov}
&&\M_T\bigl[X_{t_1}\in B_1,X_{t_2}\in B_2,\ldots,X_{t_n}\in B_n\bigr]\\ \nonumber
&:=&\int_{B_1}\int_{B_2}\cdots\int_{B_{n-1}}\int_{B_n}K_{t_1}(0,dx_1)\circ K_{t_2-t_1}(x_1,dx_2)\circ\cdots\circ K_{T-t_{n-1}}(x_{n-1},dx_n)\;.
\end{eqnarray}
(Despite the order of factors in the integrand, $x_n$ is integrated over $B_n$ first, then $x_{n-1}$ over $B_{n-1}$, etcetera.)
If $\r_0$ denotes the state at time 0,
then
   $$B\mapsto\r_0\bigl(M_t(B)[\one]\bigr)$$
is a probability measure on $\RR$, the distribution of $X_t$.
The process $X_t$ describes the cumulated observations up to time $t$.
In the particle counting example of Section \ref{BasicDavies}, $X_t$ is the number of particles counted up to time $t$.

\goodbreak\smallskip\noindent
{\bf Time-Ordered Exponential.}

\noindent
Let $f:\RR\to\L(\A)$ be a curve in the space of linear operators on $\A$.
Then by the {\it time-ordered exponential}
   $$F(s,t)=\toexp\left(\int_s^tf(u)du\right)$$
we mean the (unique) solution for $t\ge s$, if it exists, of the differential equation $\frac{\partial}{\partial t}F(s,t)=F(s,t)f(t)$
with initial condition $F(s,s)=\id_{\A}$.

\goodbreak\smallskip\noindent
{\bf Stochastic integral of a test function.}

\noindent
For the simple function $f:[0,T]\to\RR$ given by
\begin{equation}\label{SimpleFunction}
f(t):=c_j \hbox{ for }t_j\le t<t_{j+1}\;,
\end{equation}
$j=0,\ldots,n-1$,
where $0=t_0<t_1<\cdots<t_{n-1}<t_n=T$, and $c_0,\ldots,c_{n-1}\in\RR$
we define the "stochastic integral" $X(f)$ by
    $$X(f):=\sum_{j=0}^{n-1}c_j(X_{t_{j+1}}-X_{t_{j}})\;.$$

\begin{thm}\label{ThmPaths}
The instrument $\M_T$ has the characteristic function
\begin{equation}\label{BigCharFun}
\Mchat_T(f):=\EE_{\M_T}\left(e^{iX(f)}\right):=\toexp\left(\int_0^T L\bigl(f(t)\bigr)dt\right)
\end{equation}
acting on the simple functions $f:[0,T]\to\RR$ given by (\ref{SimpleFunction}).
\end{thm}

\begin{proof}
We can rewrite (\ref{EqDefMarkov}) as follows:
for Borel subsets $\tuple A n$ of $\RR$ and times  $0< t_1< t_2,\ldots,t_{n-1}<t_n=T$:
\begin{eqnarray}\label{EqProcIndep}
&&\M_T\bigl[X_{t_1}\in A_1, X_{t_2}-X_{t_1}\in A_2,\ldots,X_T-X_{t_{n-1}}\in A_n\bigr]\\\nonumber
&&\qquad\qquad\qquad=M_{t_1}(A_1)\circ M_{t_2-t_1}(A_2)\circ\cdots\circ M_{T-t_{n-1}}(A_n)\;.
\end{eqnarray}

\noindent
We may now calculate, for the simple function $f$ from (\ref{SimpleFunction}):
\begin{eqnarray*}
&&\EE_{\M_T}\left(e^{iX(f)}\right)\\
          &=&\EE_{\M_T}\Bigl(\exp i\bigl(c_0 X_{t_1}+c_1(X_{t_2}-X_{t_1})+\ldots+c_{n-1}(X_{T}-X_{t_{n-1}})\bigr)\Bigr)\\
         &=&\EE_{\M_T}\left(e^{ic_0X_{t_1}}\times e^{ic_1(X_{t_2}-X_{t_1})}\times\cdots\times e^{ic_{n-1}(X_{T}-X_{t_{n-1}})}\right)\\
         &=&\EE_{M_{t_1}}\left(e^{ic_0\bullet}\right)\circ \EE_{M_{t_2-t_1}}\left(e^{ic_1\bullet}\right)\circ\cdots\circ\EE_{M_{T-t_{n-1}}}\left(e^{ic_{n-1}\bullet}\right)\\
         &=&\Mhat_{t_1}(c_0)\circ\Mhat_{t_2-t_1}(c_1)\circ\cdots\circ\Mhat_{T-t_{n-1}}(c_{n-1})\\
         &=&e^{t_1L(c_0)}\circ e^{(t_2-t_1)L(c_1)}\circ\cdots\circ e^{(T-t_{n-1})L(c_{n-1})}\\
         &=&\toexp\left(\int_0^T L\bigl(f(t)\bigr)dt\right)\;.
\end{eqnarray*}
\end{proof}

\subsection{Discussion}
It should be noted that
this construction is not more than a skeleton, since $\F$ is a very coarse $\s$-algebra,
only fit for expressing correlations, but no path properties such as continuity and density of jumps.
The more interesting structure, however, depends highly on the choice of the generator $L$ of $(M_t)$.
If $L=L_{x,a}$, the basic Davies generator of Section \ref{BasicDavies}, then the paths will be constant, apart from jumps
of size $x$ at the random times of the quantum jumps $\psi\mapsto a\psi$.
If $L(\oo)=-\half\oo^2\cdot\id_\A$, for example, then $(X_t)$ is the Wiener process.
And if $L(\oo)=|\oo|\cdot\id_\A$, then we get a Cauchy process, and the jumps of $X_t$ are dense in time.

\noindent
By (\ref{EqProcIndep}), we may say that $X_t$ has "$\A$-dependent increments":
so formally the process $(X_t)$ looks like a {\it L\' evy process} with independent increments.
Since the increments correlate with each other via the quantum system,
we shall call this process an {\it $\A$-L\'evy process}.
(In the classical case $\A=\CC$, our term "$\A$-dependence" actually means {\it independence}: $\CC$ denotes the algebra of "no information".)

\noindent
Here we shall not attempt to go any further than this skeleton construction, and
the expression for the characteristic function of the path space measure $\M_T$ will remain completely formal.
Although the right hand side can be easily extended from the simple functions to some large class, such as $\D[0,T]$ or $C[0,T]$,
the same cannot be said about the "stochastic integral" $X(f)$ occurring on the left, to be viewed as
   $$X(f)=\int_0^T f(t)dX_t\;.$$ 
In the case of a Davies process the path $X_t$ is constant up to a finite number of jumps, and $X(f)$ is a finite sum of values of $f$.
In the case that $X_t$ is a Wiener process, it requires stochastic integration to define  $X(f)$.
In the case of the Cauchy process, the stochastic integral is quite a nontrivial matter.

\noindent
In his later work with Barchielli\cite{BarcHol}, Holevo worked out the path properties of the stochastic process $(X_t)$,
as well as the related quantum trajectories.
They found that these paths are semimartingales, which can be chosen continuous from the right with limits from the left.
The paths are continuous if and only if the term $\Lj$ in (\ref{DecompL}) is zero.


\section{Some Linear Algebra on $\L(\A)$}\label{LinAlg}
Let $\Mdop$ be the opposite algebra of $\A$: the linear space $\A$ with involution $a\mapsto a^*$ and multiplication $(a,b)\mapsto ba$.

\begin{lem}
There is a C*-algebra isomorphism $\iota$ from the algebra $\A\ten\Mdop$
to $\L(\A)$, seen as operators on $\A$ with scalar product 
    $$\inp a b:=\dth\tr(a^*b)\;,$$
which is given by
\begin{equation}\label{EqIota}
\iota(a\ten b):\quad z\mapsto azb\;.
\end{equation}
\end{lem}

\begin{proof}
By bilinearity, (\ref{EqIota}) extends to a linear map $\A\ten\Mdop\to\L(\A)$.
To show that $\iota$ is injective, suppose that, for some $\tuple a k\in\A,\tuple b k\in\Mdop$ we have
   $$\iota\left(\sok i a_i\ten b_i\right)=0\;.$$
Then for all $z\in \A=M_d$ the matrix $\sok i a_izb_i\in\A$ must be 0, in particular its $(m,n)$-matrix element.
If we require this for $z=|e_p\rangle\langle e_q|$, then we obtain, for all $m,n,p,q$
\begin{eqnarray*}
0&=&\Inp{e_m}{\left(\sok i a_izb_i\right)e_n}=\sok i\inp{e_m}{a_ie_p}\inp{e_q}{b_ie_n}\\
   &=&\Inp{e_m\ten e_q}{\left(\sok i a_i\ten b_i \right)e_p\ten e_n}\;.
\end{eqnarray*}
Therefore $\sok i a_i\ten b_i=0$, and $\iota$ is injective.
Since the dimensions of $\A\ten\Mdop$ and $\L(\A)$ are both equal to $d^4$, $\iota$ is also surjective.
In order to find $\iota(a\ten b)^*$, we calculate
\begin{eqnarray*}
\inp{\iota(a\ten b)^*y}{x}&=&\inp{y}{\iota(a\ten b)x}=\dth\tr\bigl(y^*(axb)\bigr)\\
                                          &=&\dth\tr(by^*ax)=\inp{a^*yb^*}{x}=\inp{\iota(a^*\ten b^*)y}{x}\;.
\end{eqnarray*}
So $\iota(a\ten b)^*=\iota(a^*\ten b^*)$. Generally, $\iota(u)^*=\iota(u^*)$ for $u\in \A\ten\Mdop$.
And finally, for all $a_1,a_2\in\A$ and $b_1,b_2\in \Mdop$,
   $$\iota\bigl((a_1\ten b_1)\cdot(a_2\ten b_2)\bigr)[z]=\iota(a_1a_2\ten b_2b_1)[z]= a_1a_2zb_2b_1=\iota(a_1\ten b_1)\circ\iota(a_2\ten b_2)[z]\;.$$
By linear extension, $\iota(uv)=\iota(u)\iota(v)$ for all $u,v\in \A\ten\Mdop$.
\end{proof}\noindent
For the formulation of the main theorem we shall need the following notation.
\begin{definition}
Let $\ph$ be  a state on $\A$.
Then we define the maps $c_\ph:\L(\A)\to\CC$, $l_\ph, h_\ph:\L(\A)\to \A$, and $Q_\ph:\L(\A)\to\L(\A)$ as follows:
\begin{eqnarray*}
c_\ph\bigl(\iota(u\ten v)\bigr)&:=&\ph(u)\ph(v)\;;\\
l_\ph\bigl(\iota(u\ten v)\bigr)&:=&\ph(u)\bigl(v-\ph(v)\cdot\one\bigr)\;;\\
h_\ph\bigl(\iota(u\ten v)\bigr)&:=&\frac i{2}\bigl(l_\ph\bigl(\iota(u\ten v)\bigr)-l_\ph\bigl(\iota(u\ten v)^*\bigr)\bigr)\;;\\
Q_\ph\bigl(\iota(u\ten v)\bigr)&:=&\iota\bigl(\bigl(u-\ph(u)\cdot\one\bigr)\ten\bigl(v-\ph(v)\cdot\one\bigr)\bigr)\;.
\end{eqnarray*}
\end{definition}

\noindent
We note that $Q_\ph$ has the effect of making an operator in $\L(\A)$ small when it lies close to the identity.


\section{Holevo's Theorem}\label{Holevo}
We are now in a position to formulate our central result.
Given Theorems \ref{ThmDeFinetti} and \ref{ThmPaths},
which show that weakly continuous convolution semigroups and $\A$-L\' evy processes are determined by their generators, 
the characteristic exponents, it remains to classify the latter.

\begin{thm}[Holevo 1986, modified]\label{ThmHolevo}
Let $L$ be a function $\RR\to\L(\A)$.
Then the following are equivalent:
\begin{itemize}
\item[(a)]
$L$ is the characteristic exponent of a weakly continuous convolution semigroup of instruments;
\item[(b)]
$L$ is of the form:
\begin{equation}\label{DecompL}
L(\no)[z]=i\a\no z+L_0[z]+\Ld_{\s,b}(\no)[z]+\Lj_J(\no)[z]\;,
\end{equation}
where $\a$ is a real number, $L_0$ is a Lindblad generator on $\A$, $\Ld_{\s,b}$ is a diffusive generator; i.e.:
   $$\Ld_{\s,b}(\no)[z]:=\s^2\bigl(b^*zb-\half\{b^*b,z\}+i\oo(b^*z+zb)-\half\no^2z\bigr)\;,$$
for some $\s\ge0$ and $b\in \A$;
and $\Lj_J$ is given by
\begin{eqnarray*}
\Lj_J(\no)[z]:=&{\displaystyle\intRn}&\biggl(e^{i\no x}J(dx)[z]-\half\{J(dx)[\one],z\}\\
           &&-\frac{i\no x}{1+x^2}c_\ph\bigl(J(dx)\bigr)z-i\bigl[h_\ph\bigl(J(dx)\bigr),z\bigr]\biggr)\;.
\end{eqnarray*}
where $\ph$ is any fixed state on $\A$,
and $J$ is a, possibly infinite, completely positive ``jump measure''  on $\Rp$ satisfying the conditions
\begin{eqnarray}
\label{InfDivInstCondOutcome}
\intRn\frac{x^2}{1+x^2}J(dx)&<&\infty\;;\quad\hbox{[jumps in outcome space]}\\
\label{InfDivInstCondState}
\intRn\Qphi\bigl(J(dx)\bigr)&<&\infty\;.\quad\hbox{[jumps in state space]}
\end{eqnarray}
\end{itemize}
If $J$ satisfies  these two conditions, then the integral expression for $\Lj_J$ becomes integrable due to the compensating terms containing $c_\ph$ and $h_\ph$.
\end{thm}

\noindent
Note that the conditions (\ref{InfDivInstCondOutcome}) and (\ref{InfDivInstCondState})
allow an infinity of jumps per unit time, provided that these jumps are small,
in outcome space and in state space respectively.

\noindent
In Holevo's papers\cite{Hol5} the condition (\ref{InfDivInstCondState}) had the form described in Corollary \ref{CorEquivCond} (a) below.
See Section \ref{InfDivInstIntCond} for a comparison.

\subsection{Discussion.}
Holevo's Theorem decomposes an arbitrary generator $L$ into four parts:
\begin{itemize}
\item
The shift part $L(\oo)[z]=i\a\oo z$ merely shifts the distribution $M_t$ of $X_t$ by an amount $\a t$ to the right.
On its own, it describes a needle moving to the right at speed $\a$, regardless of what happens to the quantum system.
\item
A free evolution part $L_0$, which induces a Hamiltonian or dissipative motion on the quantum system.
On its own, it describes  a quantum system that evolves according to $\r\mapsto (e^{tL_0})^*\r$,
without any connection to the needle.
\item
A diffusive part $\Ld$, where both the needle and the quantum system perform a diffusive, Brownian-like motion,
with continuous paths.
Here we may think of measurement in continuous time, of the electromagnetic field emitted by the quantum system, by heterodyne detection\cite{BarGreg09},
or continuous time measurement of some quantum observable\cite{Hol5}.
\item
A jump part $\Lj$.
Taken on its own, it describes a needle and a quantum system moving only by making jumps.
Different jump sizes of the needle are accompanied by different quantum jumps.
Infinitly many jumps may occur in a finite time, as is the case, for example, in the classical Cauchy process.\cite{Applebaum}
\end{itemize}

\noindent
There is some freedom in the representation (b) of $L$.
The state $\ph$ on $\A$ may be chosen arbitrarily, and $b$ is determined only up  to a multiple of $\one_\A$.
Moreover, there is some freedom, not used in the above formulation of the theorem, to replace the expression $\frac {x^2}{1+x^2}$ by some other function,
behaving like $x^2$ around $x=0$ and like 1 around $x=\pm\infty$.

\noindent
A change in $\ph$ will lead, via the $c_\ph$-term in $\Lj$, to a shift in the instrument, which is absorbed by the parameter $\a$,
and via the $h_\ph$-term in $\Lj$ to the addition of a term $i[h,z]$ to the generator $L$,
which is absorbed by the Hamiltonian part of $L_0$.
Replacing the function $\frac x{1+x^2}$ also leads to a finite shift in $L$, absorbed by $\a$.

\noindent
However, the variance $\s^2$ in $\Ld$, the jump measure $J$, and the non-unital part of $b$ are all completely fixed by $L$.
(Cf. the last statement in Proposition \ref{InfDivInstThmSigmaGammaDelta}.)

\noindent
It should be noted that the different components into which Holevo's formula decomposes the $\A$-L\' evy process, are highly intertwined.
Their statistical independence, which holds in the classical case $\A=\CC$, is lost.

\noindent
When performing the observations, the movements of the "needle" enable us to reconstruct the motion of the quantum system,
provided that its initial state and its generator are known to us.
In fact, the state of the quantum system can be viewed as nothing but an encoding of the observer's information concerning the quantum system,
given the past motions of the "needle".
No quantum jumps can occur without a jump of the needle.


\smallskip\noindent
{\bf Proof of Holevo's Theorem \ref{ThmHolevo}.}
Let $\L$ ("L\' evy-Khinchin-class") denote  the class of functions $L$, described by (b) of the theorem.
Our quantum version, Theorem \ref{ThmDeFinetti}, of de Finetti's theorem says that $\G=\Cbar$,
so in order to prove Holevo's Theorem \ref{ThmHolevo}, saying that $\G=\L$, it suffices to prove that $\Cbar=\L$.
We start by showing that $\Cbar\subset\L$.

\subsection{Proof that $\Cbar\subset\L$.}
Suppose that $L\in\Cbar$.
This means that $L$  is of the form (\ref{InfDivInstLApprox}); i.e., it can be written as a limit of expressions $L_n$ as in  (\ref{InfDivInstEqLfromJ})
with finite completely positive measures $J_n$ given by (\ref{InfDivInstEqJDiscr}).

\noindent
Despite the convergence $L_n\tto L$ in (\ref{InfDivInstLApprox}), the sequence of measures
$(J_n)_{n\in\NN}$ of (\ref{InfDivInstEqJDiscr}) need not itself converge.
This is the main difficulty we encounter in proving the theorem, but also the source
of interesting new contributions such as $\Ld$.

\noindent
It is gratifying that we {\it do} have convergence of three derived sequences of measures.

\begin{pro}[Convergence of $\Sigma_n$, $\Gamma_n^\ph$ and $\Delta_n^\ph$]\label{InfDivInstThmSigmaGammaDelta}
Let $\ph$ be a state on $\A$, and
let $(J_n)_{n\in\NN}$ be a sequence of finite completely positive measures on $\RR$.
Suppose that the sequence $(L_n)_{n\in\NN}$, given by (\ref{InfDivInstEqLfromJ}) converges pointwise to a continuous function $L:\RR\to\L(\A)$.
Then the sequences of completely positive measures $(\Sigma_n)_{n\in\NN}$ and $(\Gamma_n^\ph)_{n\in\NN}$, given by
\begin{eqnarray}
\label{InfDivInstEqSigman}
\Sigma_n(dx)&:=&\frac{x^2}{1+x^2}J_n(dx)\\
\label{InfDivInstEqGamman}
\hbox{and}\quad\Gamma_n^\ph(dx)&:=&\Qphi\bigl(J_n(dx)\bigr)
\end{eqnarray}
converge weakly on $\RR$ to finite completely positive measures $\Sigma$ and $\Gamma_\ph$ respectively.

\noindent
If on $\RR\setminus\{0\}$ we define the possibly infinite measure
\begin{equation}\label{EqSigmaJ}
J(dx):=\frac{1+x^2}{x^2}\Sigma(dx)\;,
\end{equation}
then we have, on $\Rp$:
\begin{equation}\label{InfDivInstEqGammaQu}
\Gamma_\ph(dx)=\Qphi\bigl(J(dx)\bigr)\;.
\end{equation}
Moreover,
the sequence $(\Delta_n^\ph)_{n\in\NN}$ of $\A$-valued measures given by
\begin{equation}\label{InfDivInstEqDeltan}
\Delta_n^\ph(dx):=\frac{x}{1+ix}l_\ph\bigl(J_n(dx)\bigr)\;,
\end{equation}
is tight and converges to an $\A$-valued measure $\Delta_\ph$ which on $\Rp$ satisfies
\begin{equation}\label{InfDivDistEqDeltaph}
\Delta_\ph(dx)=\frac{x}{1+ix}l_\ph\bigl(J(dx)\bigr)\;.
\end{equation}
In the completely positive ordering we have the inequality:
\begin{equation}\label{InfDivInstEqIneqZero}
\Delta_\ph(\{0\})^*\bullet\Delta_\ph(\{0\})\le c_\ph\bigl(\Sigma(\{0\})\bigr)\cdot\Gamma_\ph(\{0\})\;.
\end{equation}
Finally, $c_\ph\bigl(\Sigma(\{0\})\bigr)$ does in fact not depend on $\ph$, and $\Delta_\ph(\{0\})$ only up to addition of a multiple of $\one_\A$.
\end{pro}

\noindent
For the proof we shall need several steps.

\begin{proof}[Step 1: Convergence of $\Sigma_n$.]
The situation with $\Sigma_n$  is the same as in the classical L\'evy-Khinchin formula.
Consider the following sequence of integrals:
\begin{equation}\label{InfDivInstEqLConv}
F_n(\oo):=L_n(\oo)-\half\int_{-1}^1 L_n(\oo-\nu)d\nu\;.
\end{equation}
Since $L_n$ converges to $L$ uniformly on $[-1,1]$ by Theorem \ref{ThmDeFinetti},
the integrals $F_n$ converge pointwise to the continuous function $F:\RR\to\L(\A)$, given by
   $$F(\oo):=L(\oo)-\half\int_{-1}^1 L(\oo-\nu)d\nu\;.$$
Note that $F_n$ is linear in $L_n$ and vanishes if $L_n$ is constant. 

\noindent
Now, if we substitute into (\ref{InfDivInstEqLConv}) the expression (\ref{InfDivInstEqLfromJ}) for $L_n$,
the second term of the latter cancels, since it is constant as a function of $\oo$.
The first term of (\ref{InfDivInstEqLfromJ}) yields, by Fubini,
\begin{eqnarray*}
F_n(\oo)&=&\intR e^{i\oo x}J_n(dx)-\half\int_{-1}^1\left(\intR e^{i(\oo-\nu)x}J_n(dx)\right)d\nu\\
            &=&\intR\eiox\left(1-\half\int_{-1}^1 e^{-i\nu x}d\nu\right)J_n(dx)\\
             &=&\intR e^{i\oo x}\left(1-\frac{\sin x}x\right)J_n(dx)\\
            &=:&\intR e^{i\oo x}\Sigma_n'(dx)\;,
\end{eqnarray*}
where we have defined a sequence of finite completely positive measures $\Sigma'_n:=\left(1-\frac{\sin x}x\right)J_n(dx)$.
Since their characteristic functions $F_n$ converge pointwise to a continuous function $F$,
by L\'evy's Continuity Theorem \ref{ThmLevCont} $(\Sigma_n')$ is tight and converges weakly to a finite completely positive measure $\Sigma'$,
with Fourier transform $F$.
Now, note that $\Sigma'_n=k\cdot\Sigma_n$, where $k\in C(\Rbar)$ is given by
   $$k(x)=\begin{cases}\left(1-\frac{\sin x}x\right)\cdot\frac{1+x^2}{x^2}&\hbox{ if }x\notin\{0,\infty\}\\
                                    \frac16\hbox{ if }x=0;\\
                                    1\hbox{ if }x=\infty.
                                                     \end{cases}$$
As $k$ is invertible in $C(\Rbar)$, it follows that also, for all $f\in C(\Rbar)$:
\begin{eqnarray*}
\int_\Rbar f(x)\Sigma_n(dx)&=&\int_\Rbar f(x)k(x)^{-1}\Sigma'_n(dx)\\
                                         &\stackrel{n\to\infty}{\tto}&\int_\Rbar f(x)k(x)^{-1}\Sigma'(dx)=\int_\Rbar f(x)\Sigma(dx)\;.
\end{eqnarray*}
I.e., $\displaystyle\wlim_{n\to\infty}\Sigma_n=\Sigma$ on $\Rbar$, and then on $\RR$ by tightness (Lemma \ref{LemTussen}).

\noindent
{\bf Remark:}
Above we have applied a trick from Lukacs\cite{Lukacs} to deal with the convergence $F_n\tto F$: 
Since in (\ref{InfDivInstEqLConv}) the integral is over a compact set, we can apply the uniform convergence from Proposition \ref{PropWeakToChar}.
From a computational point of view, however, it would have been much smoother to write instead
    $$F_n(\oo):=L_n(\oo)-\half\intR e^{-|\nu|}L_n(\oo-\nu)d\nu\;,$$
which leads directly to $\Sigma_n$ with $\Sigmahat_n=F_n$, since $\half e^{-|\nu|}$ is the Fourier transform of $\frac{1}{1+x^2}$.
However, uniform convergence of $L_n$ to $L$ on compact sets would not have sufficed to prove pointwise convergence $F_n\tto F$.

\noindent{\bf Step 2:}
A similar construction is possible with the measure sequence $(\Gamma_n^\ph)_{n\in\NN}$.
We consider the sequence of functions $G^\ph_n:\RR\to\L(\A)$:
   $$G_n^\ph(\oo):=\Qphi\bigl(L_n(\oo)\bigr)$$
We substitute the expression (\ref{InfDivInstEqLfromJ}) for $L_n$.
Since $\Qphi$ annihilates anticommutators,
again the second term of the latter expression drops out, and we obtain that
   $$G_n^\ph(\oo)=\intR e^{i\oo x}\Qphi\bigl(J_n(dx)\bigr)=\intR e^{i\oo x}\Gamma_n^\ph(dx)\;,$$
the characteristic function of the completely positive measure $\Gamma_n^\ph$.
Clearly, $G_n^\ph(\oo)$ tends pointwise to the continuous function $\Qphi\bigl(L(\oo)\bigr)$ as $n\to\infty$.
So, again by Theorem \ref{ThmLevCont},
$\Gamma_n^\ph$ converges weakly to to a completely positive measure $\Gamma_\ph$ with Fourier transform
$\Gammahat_\ph(\oo)=\Qphi\bigl(L(\oo)\bigr)$.

\noindent
To prove (\ref{InfDivInstEqGammaQu}),
let $f\in C_b(\RR)$ have compact support contained in $\Rp$.
Then we have, as $x\mapsto f(x)\cdot(1+x^2)/x^2$ is also bounded and continuous,
\begin{eqnarray*}\label{EqJustify}
\intR f(x)\Gamma_\ph(dx)&=&\li n \int_\Rp f(x)\Qphi\bigl(J_n(dx)\bigr)\\
                       &=&\li n \int_\Rp f(x)\frac{1+x^2}{x^2}\Qphi\bigl(\Sigma_n(dx)\bigr)\\
                       &=&\intR f(x)\frac{1+x^2}{x^2}\Qphi\bigl(\Sigma(dx)\bigr)
                        = \intR f(x)\Qphi\bigl(J(dx)\bigr)\;,
\end{eqnarray*}
and the statement (\ref{InfDivInstEqGammaQu}) follows.
\end{proof}
\noindent
The measure $\Delta^\ph_n$ is connected to the pair $(\Sigma_n,\Gamma_n^\ph)$ by a kind of Cauchy-Schwarz relation.
We begin with a simple lemma.

\begin{lem}[Vector Cauchy-Schwarz]\label{InfDivInstThmVectorCS}
For $j=1,2,\ldots,n$ let $w_j\ge0$, $\a_j\in\CC$, and $\th_j\in\CC^d$. Then
   $$\Norms{\son j w_j\a_j\th_j}\le\left(\son j w_j|\a_j|^2\right)\cdot\left(\son k w_k\norms{\th_k}\right)\;.$$
\begin{proof}
Let $\th_j=(\th_{ji})_{i=1}^d$.
Then
\begin{eqnarray*}
\Norms{\son j w_j\a_j\th_j}&=&\sod i \left|\son j w_j\a_j\th_{ji}\right|^2\\
                          &\le&\sod i \left(\son j w_j|\a_j|^2\right)\cdot\left(\son k w_k|\th_{ki}|^2\right)\\
                          &=&\left(\son j w_j|\a_j|^2\right)\cdot\left(\son k w_k||\th_{k}||^2\right)\;,
\end{eqnarray*}
where the inequality is Cauchy-Schwarz in the space $l^2\bigl(\{1,2,\ldots,n\},w)$.
\end{proof}
\end{lem}

\begin{lem}[Cauchy-Schwarz for finite $\Delta_\ph$]\label{CSDelta}
Let $J$ be a finite completely positive measure on $\RR$,
and let $\Sigma(dx):=\frac{x^2}{1+x^2}J(dx)$, $\Gamma_\ph(dx):=\Qphi\bigl(J(dx)\bigr)$.
Let
   $$\Delta_\ph(dx):=\frac{x}{1+ix}l_\ph\bigl(J(dx)\bigr)\;.$$
Then we have, in the completely positive ordering on $\L(\A)$, for all $f\in C_b(\RR)$:
\begin{equation}\label{InfDivInstEqCS}
\Deltatilde_\ph(f)^*\bullet\Deltatilde_\ph(f)\le c_\ph\bigl(\Sigmatilde(|f|)\bigr)\cdot\Gammatilde_\ph(|f|)\;.
\end{equation}
\end{lem}

\begin{proof}
It suffices to prove the inequality (\ref{InfDivInstEqCS}) for $J$ of the form $J(B)[z]=\son j a_j^*z a_j\cdot\d_{x_j}(B)$.
Write $b_j:=a_j-\ph(a_j)\cdot\one$.
Then the "superoperator" $\iota(a_j^*\ten a_j):z\mapsto a_j^*za_j$ gets mapped by $c_\ph$ to $|\ph(a_j)|^2$,
by $Q_\ph$ to $z\mapsto b_j^*zb_j$, and by $l_\ph$ to $\overline{\ph(a_j)}b_j$.
Hence $c_\ph\circ\Sigma$, $\Gamma_\ph$, and $\Delta_\ph$ can be written
\begin{eqnarray}
c_\ph\bigl(\Sigmatilde(f)\bigr)&=&\son j f(x_j)\frac{x_j^2}{1+x_j^2}|\ph(a_j)|^2\;;\nonumber\\
             \Gammatilde_\ph(f)[z]&=&\son j f(x_j)b_j^*z b_j\;;\label{EqGammaExpr}\\
        \Deltatilde_\ph(f)&=&\son j f(x_j)\frac{x_j}{1+ix_j}\overline{\ph(a_j)}b_j\;.\label{EqDeltaExpr}
\end{eqnarray}
Now, let $\tuple \psi n\in\CC^d$ and $\tuple x n\in \A$ be arbitrary. 
Let $f\in\C_b(\RR)$ and $f(x_j)=|f(x_j)|e^{i\eta_j}$.
In Lemma \ref{InfDivInstThmVectorCS}\ substitute:
   $$w_j:=|f(x_j)|\;,\quad
     \a_j:=\frac{x_j}{1+ix_j}\overline{\ph(a_j)}e^{i\eta_j}\;,\quad
     \th_j:=\son i z_i b_j \psi_i\;.$$
We thus obtain
\begin{eqnarray*}
&&\Norms{\son j f(x_j)\frac{x_j}{1+ix_j}\overline{\ph(a_j)}\som i z_ib_j\psi_i}\\
          &\le&\left(\son j |f(x_j)|\frac{x_j^2}{1+x_j^2}|\ph(a_j)|^2\right)
          \cdot\left(\son k |f(x_k)|\Norms{\som i z_ib_k\psi_i}\right)\;.
\end{eqnarray*}
I.e.,
   $$\Norms{\som i z_i\Deltatilde_\ph(f)\psi_i}\le c_\ph\bigl(\Sigmatilde(|f|)\bigr)\cdot\son k |f(x_k)|\Norms{\som i z_ib_k\psi_i}\;.$$
Again put differently,
\begin{eqnarray*}
\som {i,j} \Inp{\psi_i}{\Deltatilde_\ph(f)^*z_i^*z_j\Deltatilde_\ph(f)\psi_j}
     &\le& c_\ph\bigl(\Sigmatilde(|f|)\bigr)\cdot\som {i,j}\Inp{\psi_i}{\left(\son k |f(x_k)|b_k^*z_i^*z_jb_k\right)\psi_j}\\
     &=&c_\ph\bigl(\Sigmatilde(|f|)\bigr)\cdot\som {i,j}\Inp{\psi_i}{\Gammatilde_\ph(|f|)[z_i^*z_j]\psi_j}\;,
\end{eqnarray*}
which implies (\ref{InfDivInstEqCS}).
\end{proof}

\begin{proof}[Proof of Proposition \ref{InfDivInstThmSigmaGammaDelta}, Step 3:
convergence of $(\Delta_n^\ph)$]
Let $\e>0$, and let $K>0$ be such that both
$c_\ph\bigl(\Sigma_n[-K,K]^c\bigr)$ and $\Norm{\Gamma_n^\ph[-K,K]^c[\one]}$ are less than $\e$.
Then for all $f\in C_b(\RR)$ with $\norm f\le1$ and $\supp(f)\subset[-2K, 2K]^c$ we have by Lemma \ref{CSDelta},
   $$\Deltatilde_n^\ph(f)^*\Deltatilde_n^\ph(f)\le c_\ph\bigl(\Sigma_n(|f|)\bigr)\cdot\Gammatilde_n^\ph\bigl(|f|\bigr)[\one]\;,$$
So $\Norm{\Deltatilde_n^\ph(f)}<\e$, and it follows that $(\Delta_n^\ph)$ is tight.


\noindent
In order to apply L\'evy's Continuity Theorem \ref{ThmLevCont}, it remains to show convergence of the Fourier transforms.
We follow the pattern of the case $\Sigma$:
Consider the sequence of integral expressions
\begin{equation}\label{InfDivInstEqHConv}
H_n(\oo):=l_\ph\left(L_n(\oo)-2\int_0^1 (1-\nu)L_n(\oo-\nu)d\nu\right)\;.
\end{equation}
Again, since $L_n$ converges to $L$ uniformly on $[\oo-1,\oo]$ by Theorem \ref{ThmDeFinetti},
the integrals $H_n$ converge pointwise to the corresponding continuous function $H:\RR\to\L(\A)$.
And again $H_n$ is linear in $L_n$ and vanishes if $L_n$ is constant. 
If we substitute into (\ref{InfDivInstEqHConv}) the expression (\ref{InfDivInstEqLfromJ}) for $L_n$,
the second term of the latter again cancels, since it is constant.
The first term of (\ref{InfDivInstEqLfromJ}) yields, by Fubini,
\begin{eqnarray*}
H_n(\oo)&=&l_\ph\left(\intR e^{i\oo x}J_n(dx)-2\int_{0}^1(1-\nu)\left(\intR e^{i(\oo-\nu)x}J_n(dx)\right)d\nu\right)\\
            &=&\intR e^{i\oo x}l_\ph\bigl(J_n(dx)\bigr)\kappa(x)\\
            &=:&\intR e^{i\oo x}(\Delta^\ph_n)'(dx)\;,
\end{eqnarray*}
where
   $$\kappa(x):=1-2\int_0^1(1-\nu)e^{-i\nu x}d\nu=\frac2{x^2}\biggl(e^{-ix}-(1-ix-\half x^2)\biggr)\approx\third ix+\ldots$$
and where we have defined a sequence of finite $\A$-valued measures $(\Delta_n^\ph)'(dx):=\kappa(x) l_\ph\bigl(J_n(dx)\bigr)$.
Since their characteristic functions $H_n$ converge pointwise to a continuous function $H$,
by L\'evy's Continuity Theorem \ref{ThmLevCont} there exists a finite $\A$-valued measure $\Delta_\ph'$ such that $H=\Deltahat_\ph'$ and $(\Delta_n^\ph)'\tto\Delta_\ph'$ weakly.
Now, note that
   $$(\Delta_n^\ph)'(dx)=l_\ph\bigl(J_n(dx)\bigr)\k(x)=\frac x{1+ix} l_\ph\bigl(J_n(dx)\bigr)h(x)=h(x)\Delta_n^\ph(dx)\;,$$
where $h\in C(\Rbar)$ is given by
   $$h(x)=\begin{cases}\frac{1+ix}{x}\k(x)&\hbox{ if }x\notin\{0,\infty\}\\
                                    \frac i3\hbox{ if }x=0;\\
                                    i\hbox{ if }x=\infty.
                                                     \end{cases}$$
Again $h$ is invertible in $C(\Rbar)$.
If we now {\it define} $\Delta_\ph(dx)$ on $\Rbar$ as $h(x)^{-1}\cdot\Delta'_\ph(dx)$, then it follows for all $f\in C(\Rbar)$:
\begin{eqnarray*}
\int_\Rbar f(x)\Delta_n^\ph(dx)&=&\int_\Rbar f(x)h(x)^{-1}(\Delta_n^\ph)'(dx)\\
                                                &&\stackrel{n\to\infty}{\tto}\int_\Rbar f(x)h(x)^{-1}\Delta'_\ph(dx)=\int_\Rbar f(x)\Delta_\ph(dx)\;.
\end{eqnarray*}
I.e., $\displaystyle\wlim_{n\to\infty}\Delta_n^\ph=\Delta_\ph$ on $\Rbar$, and by Lemma \ref{LemTussen} also on $\RR$.

\noindent
Equation (\ref{InfDivDistEqDeltaph}) can be proved from the weak convergence $\Sigma_n\tto\Sigma$ by the same argument as was used to prove (\ref{InfDivInstEqGammaQu}) .

\noindent
To prove the inequality (\ref{InfDivInstEqIneqZero}),
define a sequence of functions $\bigl(f_m\in\Cb(\RR)\bigr)_{k\in\NN}$ by
   $$f_m(x):=\begin{cases}
             1-m|x|&\hbox{if }|x|\le\frac1m\;;\\
             0&\hbox{otherwise}\;.
             \end{cases}$$
Applying inequality (\ref{InfDivInstEqCS}) to $f=f_m$, 
and taking $m\to\infty$, we find (\ref{InfDivInstEqIneqZero}).

\medskip\noindent{\bf Step 4: Dependence of $\s_\ph\bigl(\Sigma(\{0\})\bigr)$ and $\Delta_\ph(\{0\})$ on $\ph$.}
Let $\tau$ be the trace state $a\mapsto\frac 1d\tr(a)$ on $\A$. Then, since $f_m(x)>0$ only if $|x|<\frac1m$,
$$Q_\tau\bigl(\Sigmatilde_n(f_m)\bigr)=\intR\frac{x^2}{1+x^2}f_m(x)Q_\tau\circ J_n(dx)\le\frac1{m^2}\Gamma^\tau_n(\RR).$$
Hence
$$Q_\tau\bigl(\Sigma(\{0\})\bigr)=\li m Q_\tau\bigl(\Sigmatilde(f_m)\bigr)=\li m\li n Q_\tau\bigl(\Sigmatilde_n(f_m)\bigr)=0.$$
Now, for a completely positive map $A$ on $\A$, say $A[z]=\sok j a_j^*za_j$, to have $Q_\tau(A)=0$, means that
   $$0=Q_\tau(A)[\one]=\sok j\bigl(a_j-\ph(a_j)\cdot\one\bigr)^*\bigl(a_j-\ph(a_j)\cdot\one\bigr)\;,$$
i.e., all the $a_j$ are multiples of $\one$, and $A$ is a multiple of $\id_\A$.
We conclude that
   $$\Sigma(\{0\})=\lambda\cdot\id_\A$$
 for some $\lambda\ge0$, and therefore
 $c_\ph\bigl(\Sigma(\{0\})\bigr)=\lambda\;,$
independent of $\ph$.

\smallskip\noindent
To study the behaviour of $\Delta_\ph(\{0\})$, we may assume that $J_n$ is of the form  (\ref{InfDivInstEqJDiscr}),
so, using (\ref{EqDeltaExpr}):
   $$\Deltatilde^\ph_n(f_m)=\sokn j f_m(x_j^n)\frac{x_j^n}{1+ix_j^n}\overline{\ph(a_j^n)}\bigl(a_j^n-\ph(a_j^n)\cdot\one\bigr)\;.$$
Abbreviating $a-\ph(a)\cdot\one$ to $q_\ph(a)$ for $a\in\A$, and $q_\ph(a_j^n)$ to $b_j^n$,
we may write for any three states $\ph$, $\th$, and $\chi$ on $\A$,
using $q_\ph\circ q_\th(a)=q_\ph(a)$ and $(\th-\chi)(a)=(\th-\chi)\bigl(q_\ph(a)\bigr)$:
\begin{eqnarray*}
\Norm{q_\ph\bigl(\Deltatilde_n^\th(f_m)-\Deltatilde_n^\chi(f_m)\bigr) }&=&\Norm{\sokn j f_m(x_j^n)\frac{x_j^n}{1+ix_j^n}\overline{(\th-\chi)(a_j^n)}q_\ph(a_j^n)}\\
             &\le&\frac1m\norm{\th-\chi}\sokn j\norms{q_\ph(a_j^n)}\le\frac2m\sokn j\norm{(b_j^n)^*b_j^n}\\
             &\le&\frac2m\sokn j\tr\bigl((b_j^n)^*b_j^n\bigr)=\frac2m\tr\Gamma^\ph_n(\RR)[\one]\;.
\end{eqnarray*}
In the last step (\ref{EqGammaExpr}) was used.
It follows that
   $$q_\ph\bigl(\Delta_\th(\{0\})-\Delta_\chi(\{0\})\bigr)=\li m\li n q_\ph\bigl(\Deltatilde_n^\th(f_m)-\Deltatilde_n^\chi(f_m)\bigr)=0\;.$$
For $a\in\A$ we have that $q_\ph(a)=0$ implies $a=\ph(a)\cdot\one$.
So $\Delta_\th(\{0\})$ and $\Delta_\chi(\{0\})$ differ only by a multiple of $\one$. 
\end{proof}

We have now accomplished that we can express any $L\in\Cbar$, i.e. any limit $L$ of  the form (\ref{InfDivInstLApprox}), in terms of integrals of {\it finite} measures:

\begin{pro}[Integral Expression for a Generator]\label{InfDivInstThmFinMeasRepr}
Let $\ph$ be a state on $\A$.
Then for every $L\in\Cbar$ there are $\a\in\RR$, $h=h^*\in \A$,
completely positive measures $\Sigma$, $\Gamma_\ph$, and an $\A$-valued measure $\Delta_\ph$ such that
for all $\oo\in\RR$, $x\in \A$,
\begin{eqnarray}
L(\oo)[z]= &&\intR e^{i\oo x}\Gamma_\ph(dx)[z]-\half\{\Gamma_\ph(dx)[\one],z\}\label{InfDivInstIntOrange}\\
          &+&\intR f_\oo(x)
             c_\ph\bigl(\Sigma(dx)\bigr)\cdot z\label{InfDivInstIntBlue}\\
          &+&\intR g_\oo^{+}(x)z\Delta_\ph(dx)+g_\oo^{-}(x)\Delta_\ph(dx)^*z\label{InfDivInstIntPurple}\\
          &+&i\a\oo z\label{InfDivInstIntGreenSpace}\\
          &+&i[h,z]\label{InfDivInstIntGreenStates}\;,
\end{eqnarray}
where for $\oo\in\RR$ the bounded continuous functions $f_\oo,g^\pm_\oo:\RR\to\CC$ are given by
\begin{eqnarray*}
f_\oo(x)&:=&\begin{cases}
             \left(e^{i\oo x}-1-\frac{i\oo x}{1+x^2}\right)\frac{1+x^2}{x^2}&\hbox{if }x\ne0\\
             -\half\oo^2&\hbox{if } x=0\;;
         \end{cases}\\
g_\oo^{\pm}(x)&:=&\begin{cases}
             \left(e^{i\oo x}-1\right)\frac{1\pm ix}{x}&\hbox{if }x\ne0\\
             i\oo&\hbox{if } x=0\;,
         \end{cases}
\end{eqnarray*}
such that, on $\Rp$ the measures $\Sigma$, $\Gamma_\ph$, and $\Delta_\ph$
are connected to a single, possibly infinite measure $J$ by (\ref{EqSigmaJ}), (\ref{InfDivInstEqGammaQu}), and
(\ref{InfDivDistEqDeltaph}). On $\{0\}$ they satisfy (\ref{InfDivInstEqIneqZero}).
\end{pro}

\noindent
In the proof we shall use the following calculation repeatedly.
In fact, it was the starting point for our proof of Theorem \ref{ThmHolevo}.

\begin{lem}[Five Term Equality for Finite Measures]\label{InfDivInstEqFiveTerm}
Let $J$ be a finite completely positive measure on $\RR$.
Then for all $\oo\in\RR$ and $z\in \A$ we have the following equality for finite $\A$-valued measures:
\begin{eqnarray}
e^{i\oo x}J(dx)[z]&-&\half\{J(dx)[\one],z\}\nonumber\\
             &=&e^{i\oo x}\Qphi\bigl(J(dx)\bigr)[z]-\half\{\Qphi\bigl(J(dx)\bigr)[\one],z\}\label{InfDivInstOrange}\\
             &+&\left(e^{i\oo x}-1-\frac{i\oo x}{1+x^2}\right)c_\ph\bigl(J(dx)\bigr)\cdot z\label{InfDivInstBlue}\\
             &+&\left(e^{i\oo x}-1\right)\bigl(z l_\ph\bigl(J(dx)\bigr)+\bigl(l_\ph\bigl(J(dx)\bigr)^*z\bigr)\label{InfDivInstPurple}\\
             &+&\frac{i\oo x}{1+x^2}c_\ph\bigl(J(dx)\bigr)\cdot z\label{InfDivInstGreenSpace}\\
             &+&i\bigl[h_\ph\bigl(J(dx)\bigr),z\bigr]\label{InfDivInstGreenStates}\;.
\end{eqnarray}
\end{lem}

\begin{proof}
Consider the elementary jump generator of Proposition \ref{PropBasicDavies}:
   $$L_{x,a}(\oo)[z]=e^{i\oo x}a^*za-\half\{a^*a,z\}\;.$$
Let us define $b:=a-\ph(a)\cdot\one$, so that $a=b+\ph(a)\cdot\one$,
and substitute this on the right to find
\begin{eqnarray}\label{InfDivInstEqFiveTermab}
L_{x,a}(\oo)[z]&=&e^{i\oo x}\bigl(b+\ph(a)\cdot\one\bigr)^*z\bigl(b+\ph(a)\cdot\one\bigr)
                          -\half\bigl\{\bigl(b+\ph(a)\cdot\one\bigr)^*\bigl(b+\ph(a)\cdot\one\bigr),z\bigr\}\nonumber\\
           &=&{e^{i\oo x}b^*zb-\half\{b^*b,z\}}\nonumber\\
           &&{+\left(e^{i\oo x}-1-\frac{i\oo x}{1+x^2}\right)|\ph(a)|^2 z}\nonumber\\
           &&{+\left(e^{i\oo x}-1\right)\bigl(\overline{\ph(a)}zb+\ph(a)b^*z\bigr)}\nonumber\\
           &&{-\half\bigl[\overline{\ph(a)}b-\ph(a)b^*,z\bigr]}\nonumber\\
           &&{+\frac{i\oo x}{1+x^2}|\ph(a)|^2 z\;.}
\end{eqnarray}
Now, it suffices to prove the five term equality of Lemma \ref{InfDivInstEqFiveTerm} for $J$ of the form $\sum_{j=1}^k a_j^*\bullet a_j\cdot\d_{x_j}$.
In this lemma we then have on the left hand side:
\begin{eqnarray}\label{EqSubstDavies}
\int_B e^{i\oo x}J(dx)[z]-\half\{J(dx)[\one],z\}&=&\sum_{j:x_j\in B}e^{i\oo x_j}a_j^*za_j-\half\{a_j^*a_j,z\}\\ \nonumber
                                                 &=&\sum_{j:x_j\in B}L_{x_j,a_j}(\oo)[z]\;.
\end{eqnarray}
If again we write $b_j:=a_j-\ph(a_j)\cdot\one$,
and realize that $\Qphi(a_j^*\bullet a_j)=b_j^*\bullet b_j$, $l_\ph(a_j^*\bullet a_j)=\overline{\ph(a_j)}b_j$,
$h_\ph(a_j^*\bullet a_j)=\frac{i}{2}\bigl(\overline{\ph(a_j)}b_j-\ph(a_j)b^*\bigr)$ and $c_\ph(a_j^*\bullet a_j)=|\ph(a_j)|^2$,
then by applying (\ref{InfDivInstEqFiveTermab}) to each term of (\ref{EqSubstDavies}), and summing over $j$,
we obtain the right hand side of Lemma \ref{InfDivInstEqFiveTerm}.
\end{proof}
\goodbreak

\begin{proof}[Proof of Proposition \ref{InfDivInstThmFinMeasRepr}]
Let $L\in\Cbar$, i.e., 
$L$ is a limit of sums of elementary jump generators as in the Quantum de Finetti Theorem \ref{InfDivInstQuantDeFinetti}:
\begin{eqnarray*}
L(\oo)[z]&=&\li n L_n(\oo)[z]=\li n \sokn j L_{x_j^n,a_j^n}(\oo)[z]\\
              &=&\li n \intR e^{i\oo x}J_n(dx)[z]-\half\{J_n(dx)[\one],z\}\;,
\end{eqnarray*}
where $J_n$ is given by (\ref{InfDivInstEqJDiscr}).
We may now substitute $J_n$ for $J$ in the equality of Lemma \ref{InfDivInstEqFiveTerm},
and apply the weak convergence from Proposition \ref{InfDivInstThmSigmaGammaDelta} to the individual terms.
We obtain the following convergences:

\begin{itemize}
\item
Since the sequence of measures $\Qphi(J_n)$ converges weakly to $\Gamma_\ph$,
the integral over (\ref{InfDivInstOrange})
converges to (\ref{InfDivInstIntOrange}).
\item
Since the sequence $\Sigma_n(dx):=\frac{x^2}{1+x^2}J_n(dx)$ from (\ref{InfDivInstEqSigman}) converges weakly to $\Sigma$,
and $f_\oo$ is a bounded continuous function,
the integral over (\ref{InfDivInstBlue})
converges to (\ref{InfDivInstIntBlue}).
\item
Since the sequence $\Delta_n^\ph(dx):=\frac{x}{1+ix}l_\ph\bigl(J_n(dx)\bigr)$ from (\ref{InfDivInstEqDeltan}) converges weakly to $\Delta_\ph$,
and $g_\oo^\pm$ are bounded continuous functions,
the integral over (\ref{InfDivInstPurple})
converges to (\ref{InfDivInstIntPurple}).
\item
Since the left hand side of the equality of Lemma \ref{InfDivInstEqFiveTerm}, with $J$ replaced by $J_n$,
converges to $L(\oo)[z]$ by assumption, the sum of the last two terms (\ref{InfDivInstGreenSpace}) and (\ref{InfDivInstGreenStates})
\begin{equation}\label{InfDivInstEqGreenTerms}
\intR\frac{i\oo x}{1+x^2}c_\ph\bigl(J_n(dx)\bigr)\cdot z+i\intR\bigl[h_\ph\bigl(J_n(dx)\bigr),z\bigr]\;,
\end{equation}
must converge as well.
Putting $z=\one$, we see that the first of these must already converge on its own, say to $i\a\oo z$.
And finally, the integral over the last term must then also converge, to some derivation which we may denote by $i[h,z]$.
\end{itemize}
\end{proof}

\noindent{\bf Discussion.}
We now have expressed any $L\in\Cbar$ in terms of integrals over finite measures
$\Sigma$, $\Gamma_\ph$, and $\Delta_\ph$, which we know to be connected to an infinite measure $J$ by (\ref{EqSigmaJ}), (\ref{InfDivInstEqGammaQu}), and
(\ref{InfDivDistEqDeltaph}), and to satisfy (\ref{InfDivInstEqIneqZero}).
Although this, Proposition \ref{InfDivInstThmFinMeasRepr},
is an interesting result in itself,
the --- equivalent --- Holevo form $(b)$ of Theorem \ref{ThmHolevo}\ is much better,
since it needs only one measure, the jump measure $J$,
and highlights a singular contribution of the three measures at 0, yielding an entirely new Gauss-type contribution $\Ld$,
which is not immediately visible in Proposition \ref{InfDivInstThmFinMeasRepr}.
Our program is therefore to isolate the measures at $\{0\}$,
and to express the remaining part of the measures, the part on on $\RR\setminus\{0\}$, in terms of $J$.
This will bring the generator $L$ into the form $(b)$ of Theorem \ref{ThmHolevo}.

\begin{proof}{\bf (of the inclusion $\Cbar\subset\L$, finalization.)}
The contribution at 0, which we shall call $\Lz$, to the integral expression for $L$ in Proposition \ref{InfDivInstThmFinMeasRepr},
of the terms (\ref{InfDivInstIntOrange}), (\ref{InfDivInstIntBlue}), and (\ref{InfDivInstIntPurple}) is
\begin{eqnarray}\label{InfDivDistFormContZero}
\Lz(\oo)[z]&=&\Gamma_\ph(\{0\})[z]-\half\{\Gamma_\ph(\{0\})[\one],z\}\nonumber\\
                 &-&\half\oo^2c_\ph\bigl(\Sigma(\{0\})\bigr)+i\oo\bigl(\Delta_\ph(\{0\})^*z+z\Delta_\ph(\{0\})\bigr)\;.
\end{eqnarray}
Let us define $\s\ge0$ and $b\in \A$ by
\begin{eqnarray}\label{EqDefSigma}
\s^2&:=&c_\ph\bigl(\Sigma(\{0\})\bigr)\;;\\
b&:=&\begin{cases}\label{EqDefB}
         \frac1{\s^2}\Delta_\ph(\{0\})&\hbox{ if }\s^2>0\;;\\
         0&\hbox{ if }\s^2=0\;.
     \end{cases}\\\nonumber
\end{eqnarray}
(Note that, by the Cauchy-Schwarz inequality (\ref{InfDivInstEqIneqZero}), $\s^2=0$ implies that $\Delta_\ph(\{0\})$ vanishes, anyway.)
Then by the same inequality we have $\s^2\Gamma_\ph(\{0\})\ge \s^4b^*\bullet b$ in the completely positive ordering,
hence we obtain a completely positive map $B:\A\to \A$ by putting
$B[z]:=\Gamma_\ph(\{0\})[z]-\s^2b^*zb$, so that
   $$\Gamma_\ph(\{0\})[z]=B[z]+\s^2b^*zb\;.$$
Then let $\Ltilde_0[z]:=B[z]-\half\{B[\one],z\}$, a Lindblad generator.
We obtain:
\begin{eqnarray*}
\Gamma_\ph(\{0\})[z]-\half\{\Gamma_\ph(\{0\})[\one],z\}&=&B[z]+\s^2 b^*zb-\half\{B[\one]+\s^2 b^*b,z\}\\
                               &=&\Ltilde_0[z]+\s^2\bigl(b^*zb-\half\{b^*b,z\}\bigr)\;.
\end{eqnarray*}
We have thus brought the contribution (\ref{InfDivDistFormContZero}) into the form
   $$\Lz(\oo)=\Ltilde_0[z]+\s^2\bigl(b^*zb-\half\{b^*b,z\}+i\oo(b^*z+zb)-\half\oo^2z\bigr)=\Ltilde_0[z]+\Ld_{\s,b}(\oo)[z]\;,$$
and have established the diffusive term $\Ld$ in Theorem \ref{ThmHolevo}.

\noindent
In order to identify the part $\Lj$ in $(b)$ of the theorem,
we introduce the measure $J_\e(dx):=1_{\RR\setminus(-\e,\e)}\cdot J(dx)$.
As $\frac{1+x^2}{x^2}$ is bounded on $\RR\setminus(-\e,\e)$, this measure $J_\e$ is finite,
and we may apply the five term equality of Lemma \ref{InfDivInstEqFiveTerm}\ to it.
Now bring the terms (\ref{InfDivInstGreenSpace}) and (\ref{InfDivInstGreenStates}) to the other side of the equation in this lemma,
and integrate both the left and the right hand side over $\RR\setminus(-\e,\e)$.
Then on the right hand side only the finite measures (\ref{InfDivInstOrange}),   (\ref{InfDivInstBlue}), and  (\ref{InfDivInstPurple}) remain,
so we may take the limit $\e\downarrow0$, to find their integral over $\RR\setminus\{0\}$.
It follows that also the left hand sided converges as $\e\downarrow0$.
We obtain the equality of convergent integrals
\begin{eqnarray}\label{EqConvInt}
&&\intRn\biggl(\eiox J(dx)[z]-\ant{J(dx)}{z}-\frac{i\oo x}{1+x^2}c_\ph\bigl(J(dx)\bigr)\cdot z - i\bigl[h_\ph\bigl(J(dx)\bigr),z\bigr]\biggr)\nonumber\\
&=&\intRn Q_\ph\bigl(J(dx)\bigr)[z]-\half\ant{Q_\ph(dx)}{z}\nonumber\\
&&+\intRn\left(\eiox-1-\frac{i\oo x}{1+x^2}\right)c_\ph\bigl(J(dx)\bigr)\cdot z\nonumber\\
&&+\intRn\bigl(\eiox-1\bigr)\biggl(z l_\ph\bigl(J(dx)\bigr)+l_\ph\bigl(J(dx)\bigr)^*z\biggr)\;.
\end{eqnarray}
Now in Proposition \ref{InfDivInstThmFinMeasRepr},
bring the contribution $\Lz(\oo)[z]=\Ltilde_0[z]+\Ld_{\s,b}(\oo)[z]$ to $L$ at 0, as well as the terms $i\a\oo z$ and $i[h,z]$ to the left: 
\begin{eqnarray}\label{EqFiniteRight}
L(\oo)[z]&-&i\a\oo x-i[h,z]-\Ltilde_0[z]-\Ld_{\s,b}(\oo)[z]\nonumber\\
          &=&\intRp\left(\eiox\Gamma_\ph(dx)[z]-\half\{\Gamma_\ph(dx)[\one],z\}\right)\nonumber\\
          &&+\intRp f_\oo(x)c_\ph\bigl(\Sigma(dx)\bigr)\cdot z \nonumber\\
          &&+\intRp \bigl(g_\oo^{+}(x)z\Delta_\ph(dx)+g_\oo^{-}(x)\Delta_\ph(dx)^*z\bigr)\;.
\end{eqnarray}
By the equalities (\ref{EqSigmaJ}), (\ref{InfDivInstEqGammaQu}), and (\ref{InfDivDistEqDeltaph}), (\ref{EqFiniteRight}) equals the right hand side of (\ref{EqConvInt}).
So (\ref{EqFiniteRight}) also equals the left hand side of  (\ref{EqConvInt}).
If we put $L_0[z]:=\Ltilde_0[z]+i[h,z]$, then $L$ obtains the form (b) of Theorem \ref{ThmHolevo}. In other words: $L\in\L$.
We have shown that $\Cbar\subset\L$.
\end{proof}

\subsection{Proof that $\L\subset\Cbar$}
We show that each of the terms in the L\'evy-Khinchin form $(b)$ of Theorem \ref{ThmHolevo} can be approximated by sums of elementary generators.
For the free terms $i\a\oo z$, $i[h,z]$, $\Ltilde_0[z]$, and the diffusive generator $\Ld_{\s,b}(\oo)[z]$ this can be done explicitly. 

\begin{lem}\label{InfDivInstEqApproxFree}
For all $\oo\in\RR$ and $z\in \A$ we have
\begin{eqnarray*}
i\a\oo z&=&\li n n L_{\a/n,\one}(\oo)[z]\;;\\
i[h,z]&=&\li n n L_{0,u_n}(\oo)[z],\quad\hbox{where} \quad u_n:=e^{-ih/n}\;;\\
\Ltilde_0&=&\sum_j L_{0,a_j}(\oo)[z]\quad\hbox{for some}\quad \tuple a k\in \A\;;\\
\Ld_{\s,b}(\oo)&=\half\s^2&\li n n^2\left(L_{\frac1n,\one+\frac1n b}(\oo)[z]+L_{-\frac1n,\one-\frac1n b}(\oo)[z]\right)\;.
\end{eqnarray*}
\end{lem}
\begin{proof}
These are straightforward verifications.
\end{proof}

\noindent
Also the last term $\Lj_J$ in $(b)$ of Theorem \ref{ThmHolevo}\ lies in $\Cbar$:
\begin{lem}
If $J$ is a completely positive measure on $\RR$, satisfying the conditions (\ref{InfDivInstCondOutcome}) and (\ref{InfDivInstCondState}),
then the integral defining $\Lj_J$ in Theorem \ref{ThmHolevo} converges,
and $\Lj_J$ can be approximated pointwise by sums of elementary generators.
\end{lem} 

\begin{proof}
For $\e>0$, let $J_\e$ denote the finite restriction of $J$ to $\RR\setminus(-\e,\e)$.
Consider the function $\Lj_{J_\e}$, which can be written
   $$\Lj_{J_\e}(\oo)=A_\e(\oo)-B_\e(\oo)\;,$$
where $A_\e,B_\e:\RR\to\L(\A)$ are given by
\begin{eqnarray*}
A_\e(\oo)[z]&:=&\int_{\RR\setminus(-\e,\e)}\biggl(e^{i\oo x}J(dx)[z]-\{J(dx)[\one],z\}\biggr)\quad\hbox{and}\\
B_\e(\oo)[z]&:=&\int_{\RR\setminus(-\e,\e)}\left(\frac{i\oo x}{1+x^2}c_\ph\bigl(J(dx)\bigr)\cdot z+i[h_\ph\bigl(J(dx)\bigr),z]\right)=:i\a_\e\oo z+i[h_\e,z]\;.
\end{eqnarray*}
Applying the five term equality Lemma \ref{InfDivInstEqFiveTerm} to the finite measure $J_\e$,
we see that $\Lj_{J_\e}$ is given by the sum of the integrals of the terms
(\ref{InfDivInstOrange}), (\ref{InfDivInstBlue}), and (\ref{InfDivInstPurple}), with $J$ replaced by $J_\e$.
Now, by assumptions (\ref{InfDivInstCondOutcome}) and (\ref{InfDivInstCondState})
the terms (\ref{InfDivInstOrange}) and (\ref{InfDivInstBlue}) converge as $\e\downarrow0$.
By Cauchy-Schwarz, (\ref{InfDivInstEqCS}), so does (\ref{InfDivInstPurple}) 
Therefore also, for all $\oo\in\RR$,
    $$\lim_{\e\downarrow0}\Lj_{J_\e}(\oo)=\Lj_J(\oo)\;.$$
Now, by approximating $J_\e$ with measures of the form $\sum_{j=1}^ka_j^*\bullet a_j\cdot\d_{x_j}$,
we can approximate $A_\e$ pointwise with elements of $\C$, so that $A_\e\in\Cbar$.
By Lemma \ref{InfDivInstEqApproxFree}, we have $B_\e\in\Cbar$.
Therefore $\Lj_{J_\e}\in\Cbar$, and by the above also $\Lj_J\in\Cbar$.
\end{proof}


\section{Equivalence of Integrability Conditions}\label{InfDivInstIntCond}
In our formulations of Holevo's Theorem \ref{ThmHolevo} two conditions, (\ref{InfDivInstCondOutcome}) and (\ref{InfDivInstCondState}), 
were imposed on the L\'evy measure $J$ concerning jumps in observation space and jumps in state space respectively.
Also in Holevo's paper \cite{Hol5} two conditions were imposed. 
Our second condition seems to differ considerably from Holevo's. See Corollary \ref{CorEquivCond} (a).
Here we shall show that, despite appearances, the conditions are actually equivalent.
We first compare the two terminologies:

\begin{pro}\label{ContQuantTraThmEqCond}
Let $A_1,A_2,\ldots$ be an increasing sequence of completely positive maps on $\A$,
and let $\ph$ be a state on $\A$. Then the following are equivalent:

\begin{itemize}
\item[(a)]
For all $\tuple\psi m \in\CC^d$ and $\tuple z m \in \A$ satisfying $\sum_{i=1}^m z_i\psi_i=0$, we have
   $$\li n \som i\som j\inp{\psi_i}{A_n[z_i^*z_j]\psi_j}<\infty\;;$$
\item[(b)]
$\displaystyle\li n \norm{Q_\ph(A_n)[\one]}<\infty$.
\end{itemize}
\end{pro}

\begin{proof}
$(b)\Longrightarrow(a)$:
First let $A:\A\to \A:z\mapsto a^*za$ for some $a\in \A$.
Then we have, for $\tuple\psi m$ and $\tuple z m$ as in $(a)$:
\begin{eqnarray}\label{ContQuantTraQcalc}
\som i\som j\inp{\psi_i}{Q_\ph(A)[z_i^*z_j]\psi_j}&=&\som i\som j\inp{\psi_i}{\Pfi{a}^* z_i^*z_j \Pfi{a}\psi_j}\nonumber\\
   &=&\Norms{\som j z_j\Pfi{a}\psi_j}
    =\Norms{\som j z_j a \psi_j}\nonumber\\
   &=&\som i\som j\inp{\psi_i}{A[z_i^*z_j]\psi_j}\;.
\end{eqnarray}
By linearity, the equality between the left and the right hand sides extends to all $A\in\CP(\A)$.
Hence the limit expression in $(a)$ equals the same expression with $A_n$ replaced by $Q_\ph(A_n)$,
which by $(b)$ is bounded above for all choices of $\tuple\psi m$ and $\tuple z m$.
It follows that the limit in $(a)$ is finite.

\noindent
$(a)\Longrightarrow(b)$:
Let $\tuple\psi d$ be an orthonormal basis of eigenvectors of the density matrix of $\ph$,
so that we may write
   $$\ph(a)=\sod i r_i\inp{\psi_i}{a\psi_i}$$
for some $\tuple r d\ge0$.
Now let for $i,j=1,\ldots d$ the matrix $z_j^{(i)}\in\A$ be given by
   $$z_j^{(i)}:=\d_{ij}\cdot\one-r_j\ket{\psi_i}\bra{\psi_j}\;.$$
Then, for $a\in \A$,
  $$\sod j z_j^{(i)} a\psi_j=\sod j\bigl(\d_{ij}\cdot\one-r_j\ket{\psi_i}\bra{\psi_j}\bigr)a\psi_j
                          =\Pfi{a}\psi_i\;.$$
Hence, if $A[z]=a^*za$, we find, applying part of (\ref{ContQuantTraQcalc}):
\begin{eqnarray*}
\sod i\left(\sod j\sod k\inp{\psi_j}{A[(z_j^{(i)})^*z_k^{(i)}]\psi_k}\right)&=&\sod i\Norms{\sod j z_j^{(i)}a\psi_j}
                   =\sod i\Norms{\Pfi{a}\psi_i}\\
                  &=&\sod i\inp{\psi_i}{\Pfi{a}^*\Pfi{a}\psi_i}\\
                  &=&\sod i\inp{\psi_i}{Q_\ph(A)[\one]\psi_i}
                   =\tr\bigl(Q_\ph(A)[\one]\bigr)\;.
\end{eqnarray*}
Again the equality of the left and the right hand sides extends to all $A\in\CP(\A)$.
If we now replace $A$ by $A_n$, and let $n$ tend to infinity, then the left hand side remains bounded by assumption $(a)$.
Hence so does the right hand side. The conclusion follows since $\norm{c}\le\tr(c)$ for all positive $c\in \A$.
\end{proof}

\begin{cor}\label{CorEquivCond}
Let $J$ be a completely positive measure on $\RR\setminus\{0\}$ satisfying
\begin{equation}\label{ContQuantTraEqJumpCond}
\int_{\RR\setminus\{0\}}\frac{x^2}{1+x^2}J(dx)<\infty\;,
\end{equation}
and let $\ph$ be a state on $\A$.
Then the following conditions are equivalent.
\begin{itemize}
\item[(a)]{\bf (Holevo's second condition)}
For all $\tupple\psi m\in\CC^d$ and $\tupple z m\in \A$ satisfying $\sod i z_i\psi_i=0$:
   $$\som i\som j\int_{\RR\setminus\{0\}}\inp{\psi_i}{J(dx)[z_i^*z_j]\psi_k}<\infty\;;$$
\item[(b)]{\bf (The jump condition (\ref{InfDivInstCondState}) in state space)}:
$$\intRn Q_\ph\bigl(J(dx)\bigr)<\infty\;.$$
\end{itemize}
\end{cor}

\begin{proof}
By (\ref{ContQuantTraEqJumpCond}), for all $n\ge1$ the CP map
   $$A_n:=J\left(\RR\setminus\left(-\frac1n,\frac1n\right)\right)$$
is well-defined. Clearly $A_1\le A_2\le\ldots$ in the completely positive ordering.
Then the statement follows from Proposition \ref{ContQuantTraThmEqCond}.
\end{proof}


\section*{Acknowledgments}
This paper grew from a chapter in the book\cite{Buch} on Quantum Markov Processes
that Burkhard K\"ummerer and the author are preparing.
I thank Burkhard and his wife Andrea for their kind hospitality, streching over many years, in their home in Seeheim, Germany.

\noindent
I also wish to thank my cousin, the concert pianist Ellen Corver, who expressed her fondness of science by lending me her holiday house in the Marche, Italy, for the preparation of this work.


\end{document}